\documentclass[11pt,letterpaper]{article}

\usepackage[margin=1in]{geometry}

\usepackage{mathpazo}

\usepackage{enumerate}

\usepackage{microtype}

\usepackage{tikz}
\usetikzlibrary{calc}

\usepackage{epstopdf}

\usepackage[colorlinks = true]{hyperref}
\hypersetup{
  pdftitle = {We're looking for jobs!}
}
\definecolor{darkred}  {rgb}{0.5,0,0}
\definecolor{darkblue} {rgb}{0,0,0.5}
\definecolor{darkgreen}{rgb}{0,0.5,0}
\hypersetup{
  urlcolor   = blue,         
  linkcolor  = darkblue,     
  citecolor  = darkgreen,    
  filecolor  = darkred       
}

\usepackage{amsmath,amssymb,amsfonts,amsthm,amstext}

\usepackage{etoolbox}

\usepackage{mathtools}
\mathtoolsset{centercolon}
\makeatletter
\protected\def\tikz@nonactivecolon{\ifmmode\mathrel{\mathop\ordinarycolon}\else:\fi}
\makeatother

\usepackage{cleveref}
\crefname{lemma}{Lemma}{Lemmas}
\crefname{proposition}{Proposition}{Propositions}
\crefname{definition}{Definition}{Definitions}
\crefname{theorem}{Theorem}{Theorems}
\crefname{conjecture}{Conjecture}{Conjectures}
\crefname{corollary}{Corollary}{Corollaries}
\crefname{section}{Section}{Sections}
\crefname{appendix}{Appendix}{Appendices}
\crefname{figure}{Fig.}{Figs.}
\crefname{equation}{Eq.}{Eqs.}
\crefname{table}{Table}{Tables}
\crefname{claim}{Claim}{Claims}
\crefname{example}{Example}{Examples}

\newcommand{\ket}[1]{|#1\rangle}
\newcommand{\bra}[1]{\langle#1|}
\newcommand{\braket}[2]{\langle#1|#2\rangle}
\newcommand{\ketbra}[2]{|#1\rangle\!\langle#2|}
\newcommand{\proj}[1]{|#1\rangle\!\langle#1|}

\newcommand{\x}{\otimes}
\newcommand{\xp}[1]{^{\otimes #1}}

\newcommand{\ct}{^{\dagger}}

\DeclarePairedDelimiter{\set}{\lbrace}{\rbrace}
\DeclarePairedDelimiter{\abs}{\lvert}{\rvert}
\DeclarePairedDelimiter{\norm}{\lVert}{\rVert}
\DeclarePairedDelimiter{\of}{\lparen}{\rparen}
\DeclarePairedDelimiter{\sof}{\lbrack}{\rbrack}

\renewcommand{\Re}{\operatorname{Re}}
\renewcommand{\Im}{\operatorname{Im}}

\DeclareMathOperator{\spn}{span}

\DeclareMathOperator{\Tr}{Tr}
\DeclareMathOperator{\ad}{ad}
\DeclareMathOperator{\poly}{poly}

\newcommand{\C}{\mathbb{C}}
\newcommand{\R}{\mathbb{R}}

\newcommand{\HS}{\mathcal{H}}
\newcommand{\A}{\mathsf{A}}
\newcommand{\B}{\mathsf{B}}

\newcommand{\Herm}[1]{\mathrm{Herm}(#1)} 
\newcommand{\D}[1]{\mathrm{D}(#1)} 
\newcommand{\U}[1]{\mathrm{U}(#1)}

\newcommand{\1}{\mathbb{1}}

\newtheorem{theorem}{Theorem}
\newtheorem{lemma}[theorem]{Lemma}

\newtheorem{corollary}[theorem]{Corollary}

\theoremstyle{definition}
\newtheorem{remark}[theorem]{Remark}

\usepackage{authblk}

\begin{document}

\title{Hamiltonian Simulation with Optimal Sample Complexity}



\author[1]{Shelby Kimmel\thanks{All authors contributed equally to this work; the ordering is alphabetical. E-mails of authors in order: 
\texttt{shelbyk@umd.edu}, \texttt{cedricl@umiacs.umd.edu}, \texttt{glow@mit.edu}, 
\texttt{marozols@gmail.com}, \texttt{tjyoder@mit.edu}}}
\author[1]{Cedric Yen-Yu Lin}
\author[2]{Guang Hao Low}
\author[3]{Maris Ozols}
\author[2]{Theodore J.~Yoder}

\affil[1]{Joint Center for Quantum Information and Computer Science (QuICS), University of Maryland}
\affil[2]{Department of Physics, Massachusetts Institute of Technology}
\affil[3]{Department of Applied Mathematics and Theoretical Physics, University of Cambridge}
\date{}
\maketitle

\begin{abstract}
We investigate the sample complexity of Hamiltonian simulation: how many copies of an unknown quantum state are required to simulate a Hamiltonian encoded by the density matrix of that state? We show that the procedure proposed by Lloyd, Mohseni, and Rebentrost [\href{https://dx.doi.org/10.1038/nphys3029}{{\em Nat.\ Phys.}, 10(9):631--633, 2014}] is optimal for this task. We further extend their method to the case of multiple input states, showing how to simulate any Hermitian polynomial of the states provided. As applications, we derive optimal algorithms for commutator simulation and orthogonality testing, and we give a protocol for creating a coherent superposition of pure states, when given sample access to those states. We also show that this sample-based Hamiltonian simulation can be used as the basis of a universal model of quantum computation that requires only partial swap operations and simple single-qubit states.
\end{abstract}


\section{Introduction}

One of the most anticipated applications of quantum computation is Hamiltonian
simulation. In fact, this was Feynman's main motivation for suggesting the
creation of quantum computers~\cite{feynman1982simulating}. Feynman's intuition
may soon pay off, as the simulation of Hamiltonians for quantum chemistry looks
to be implementable on small to moderately sized quantum computers
\cite{hastings2015improving,wecker2014gate}. In addition, Hamiltonian simulation
has implications for more general computational problems, including adiabatic
algorithms~\cite{farhi2000quantum}, quantum walk
algorithms~\cite{childs2003exponential}, and algorithms for systems of linear
equations ~\cite{harrow2009quantum}.

Much work has been done on the \emph{time} and \emph{query} complexity of
Hamiltonian simulation when given a classical description or black box
description of the Hamiltonian. Lloyd provided the first formal results on
simulation, considering Hamiltonians that consist of sums of non-commuting
terms~\cite{lloyd1996universal}. Other lines of research have focused on simulating
sparse Hamiltonians, with a long sequence of work recently culminating in an optimal algorithm~\cite{LC16}
(see~\cite{berry2015hamiltonian} for a more complete history of work in this field).

In this work, we approach the problem of Hamiltonian simulation from a
slightly different perspective. Rather than given a classical description or black-box
access to a Hamiltonian $H$, we consider the problem of simulating $H$
when given many copies of a quantum state $\rho$ that encodes
the Hamiltonian to be simulated. In particular, we  assume that
\begin{equation}
  \rho = \frac{H + c \1}{\Tr (H + c \1)}
\end{equation}
for some constant $c \in \R$ such that $H + c \1$ is positive semidefinite. 
In that case, $\rho$ itself is positive semidefinite and $\Tr \rho = 1$,
so $\rho$ is a valid
density matrix.  Note that 
\begin{equation}
  e^{-i\rho t}
  = \exp \of*{\frac{-i (H + c \1) t}{\Tr (H + c \1)}}
  = \exp \of*{\frac{-ict}{\Tr (H + c \1)}}
    \exp \of*{-iH \frac{t}{\Tr (H + c \1)}},
  \label{eq:exprho}
\end{equation}
so the Hamiltonian dynamics of $H$ and $\rho$ are equivalent up to an overall phase and time scaling.
Moreover, since the Hamiltonian $H$ in \cref{eq:exprho} can be arbitrary, any unitary can in fact be expressed as $e^{-i \rho t}$ for an appropriately chosen state $\rho$ and time $t$.

This modified version of the original Hamiltonian simulation problem is what we call \emph{sample-based} Hamiltonian simulation: given one copy of an unknown state $\sigma$ and $n$
copies of an unknown state $\rho$, implement the following map:
\begin{equation}
  \sigma \x \underbrace{\rho \x \dotsb \x \rho}_{n}
  \quad\mapsto\quad
  e^{-i \rho t} \sigma e^{i \rho t}
  \label{eq:map}
\end{equation}
where $t$ is the desired evolution time. We also
allow for some error in the final state---we denote by $\delta$ the \emph{trace distance}~\cite{NC10} between
the state that is output by the protocol and the ideal state $e^{-i \rho t} \sigma e^{i \rho t}$.
This problem was first considered in~\cite{LMR14}, where the authors give a simple protocol, which
we call the \emph{LMR protocol} (after the authors' initials), for
approximately implementing the unitary $e^{-i\rho t}$ using many copies of $\rho$.
Their protocol is based on a \emph{partial swap} operation that can also be considered as a finite-dimensional analogue of a beam-splitter \cite{ADO15}.
An interesting feature of the LMR protocol is that it is agnostic with regard to $\rho$. In the spirit of \cite{Preskill, GC99}, this suggests interpreting $\rho$ as a ``quantum software state''.

The main motivation for sample-based Hamiltonian simulation in \cite{LMR14} is to perform principal component analysis of $\rho$. They do this by performing phase estimation on the unitary $e^{-i\rho}$. We note that it is a nontrivial fact that the controlled-$e^{-i\rho}$ operation can in fact be implemented with the LMR protocol; see \cref{sec:control-LMR} (this fact does not seem to have been explicitly discussed in previous work). In addition, we note in \cref{sec:phase_est} that a slightly more careful analysis gives a polynomial improvement in sample complexity over the complexity given in \cite{LMR14} for performing phase estimation. The LMR protocol has applications to many problems in machine learning, e.g. \cite{LMR14,Wang14,RML14,CL15}.

In this paper, we investigate the optimal scaling of sample-based Hamiltonian simulation. That is, we ask the following question: given $t$ and $\delta$, what is
the minimum $n$ (number of copies of $\rho$) necessary to implement
$e^{-i\rho t}$ on an unknown state $\sigma$ to trace distance at most $\delta$?
We call this the \emph{sample complexity} of Hamiltonian simulation. 

It is interesting to consider alternative strategies, other than LMR, for sample-based Hamiltonian simulation.
While the LMR protocol acts with each copy of $\rho$ sequentially,
perhaps one could achieve better performance by acting with a global operation \cite{Ozols15}?
For example, recent near-optimal tomographic protocols have relied on performing
global operations (like the Schur transform) on many copies of the unknown state
\cite{HHJ+15,o2015efficient}. Along those lines, perhaps one could do better than LMR by applying
tomographic protocols to get an estimate $\widehat{\rho}$ of $\rho$ from the $n$ copies
of $\rho$, and then evolve according  to $e^{-i \widehat{\rho} t}$. 

On the contrary, however, we show that the LMR protocol performs sample-based Hamiltonian simulation with
asymptotic optimality in both $t$ and $\delta$ simultaneously (\cref{sec:optim}). In fact, LMR performs asymptotically 
\emph{better} than any tomographic strategy (\cref{sec:Prelim}) for sample-based Hamiltonian simulation. While the 
lower bound of \cref{sec:optim} uses mixed states, in \cref{sec:optimal}, we show a
matching lower bound in $\delta$ even when restricting to pure states. In the process, we provide a sample-optimal algorithm for 
a variant of Grover's search. In \cref{comm_sim}, we discuss the sample complexity of more complex Hamiltonians 
that depend on multiple states. For example, we show how to simulate the Hamiltonians given by the commutator
$i[\rho_1,\rho_2]$ and anti-commutator $\{\rho_1,\rho_2\}$ of two states $\rho_1$ and $\rho_2$, when given access to many copies of those states, and prove that our protocol 
is optimal. We also 
show how to simulate any real linear combination of states $\rho_1,\dotsc,\rho_K$ and, by combining these observations, any Hermitian polynomial (i.e.\ any element of the Jordan-Lie algebra~\cite{Emch}, see \cref{sect:Jordan}) generated by these states, when given 
access to many copies of those states. In \cref{sec:opt_comm_sim}, we give applications of commutator simulation to orthogonality testing and quantum state addition. In \cref{sec:universal}, we show how to use sample-based 
Hamiltonian simulation to implement a universal model of quantum computation using only partial swaps 
and a stream of input qubits initialized in $\ket{0}$ and $\ket{+}$. Finally, in \cref{sec:outlook}, we discuss some open questions.

\subsection*{Notation}

We use $\HS$ to denote a finite-dimensional Hilbert space, and $\D{\HS}$ to represent the set of positive semi-definite
operators with trace $1$ in $\HS$ (i.e.\ the set of valid quantum states).

The \emph{trace distance} between $\rho, \sigma \in \D{\HS}$ is given by
\begin{equation}
\frac{1}{2} \norm{\rho-\sigma}_1
\end{equation}
where $\|A\|_1:=\Tr(\sqrt{A A^\dagger}).$ The trace distance between  $\rho$ and
$\sigma$ gives the maximum  difference in probability of any measurement on the
two states~\cite{NC10}.
For two quantum channels $\mathcal{E}_1$ and $\mathcal{E}_2$ that act on $\D{\HS}$, their \emph{trace norm distance} is defined as
\begin{equation}
\frac{1}{2}\|\mathcal{E}_1-\mathcal{E}_2\|_{\mathrm{tr}} := \frac{1}{2}\max_{\rho\in \D{\HS}} \|\mathcal{E}_1(\rho)-\mathcal{E}_2(\rho)\|_1 
\end{equation}
 The \emph{diamond norm distance} is defined as
\begin{equation}
\frac{1}{2}\|\mathcal{E}_1-\mathcal{E}_2\|_{\diamond} := \frac{1}{2}\max_{k,\rho\in \D{\HS\otimes \HS_k}} \norm[\big]{(\mathcal{E}_1 \otimes \mathcal{I})(\rho)-(\mathcal{E}_2 \otimes \mathcal{I})(\rho)}_1
\end{equation}
where $\mathcal{I}$ is the identity channel on a $k$-dimensional space $\HS_k$. By definition, $\|\mathcal{E}_1-\mathcal{E}_2\|_{\diamond} \ge \|\mathcal{E}_1-\mathcal{E}_2\|_{\mathrm{tr}}$. For
unitary channels $\mathcal U_1$ and $\mathcal U_2$ corresponding to conjugation by unitary matrices $U_1$ and $U_2$, we will sometimes write 
$\|U_1-U_2\|_{\diamond}$ to mean $\|\mathcal U_1-\mathcal U_2\|_{\diamond}$. 
If $\rho_{\A\B}\in \D{\HS_\A\otimes \HS_\B}$, by $\Tr_\B\rho_{\A\B}$ we mean taking the partial trace of the second system. More generally, given a state $\rho$ on multiple subsystems, by $\Tr_i\rho$ we mean taking the partial trace of the $i^{\mathrm{th}}$ subsystem of $\rho.$

We define the single qubit state $\ket{+}\coloneqq(\ket{0}+\ket{1})/\sqrt{2}$.
We refer to $i[A,B]:=i(AB-BA)$ as the \emph{commutator} and $\{A,B\} := AB+BA$ as the \emph{anticommutator} of operators
$A$ and $B$. We will use $X$, $Y$, and $Z$ to denote the single-qubit Pauli operators.
We use $\1_\A$ to mean the identity matrix acting on subsystem $\A$, but if clear from context, we will drop the subscript.

\section{LMR Protocol versus State Tomography} \label{sec:Prelim}

Lloyd, Mohseni, and Rebentrost~\cite{LMR14} gave a simple method for approximating the transformation in \cref{eq:map}. Importantly, their procedure is independent of $\sigma$ and $\rho$, and the number of copies of $\rho$ required does not depend on the dimension of the two states. We state their result in a slightly more general form, where $\sigma$ has two registers and $e^{-i \rho t}$ is applied only to one of them.

\begin{theorem}[\cite{LMR14}]\label{thm:LMR}
Let $\rho \in \D{\HS_\A}$ and $\sigma \in \D{\HS_\A \x \HS_\B}$ be two unknown quantum states and $t \in \R$ (can be either positive or negative). Then there exists a quantum algorithm that transforms $\sigma_{\A\B} \x \rho_{\A_1} \x \dotsb \x \rho_{\A_n}$ into $\tilde{\sigma}_{\A\B}$ such that
\begin{equation}
  \frac{1}{2}\norm[\big]{
    \of[\big]{e^{-i \rho_\A t} \x \1_\B} \sigma_{\A\B} \of[\big]{e^{i \rho_\A t} \x \1_\B}
    -\tilde{\sigma}_{\A\B}
  }_1 \leq \delta,
\end{equation}
as long as the number of copies of $\rho$ is $n = O(t^2 / \delta)$. In other words, this quantum algorithm implements the unitary $e^{-i\rho t}$ up to error $\delta$ in diamond norm, using $O(t^2/\delta)$ copies of $\rho$.
\end{theorem}
\begin{proof}
We will give a sketch of the proof; for the full proof see \cref{apx:LMR}. For simplicity we assume $\rho$ and $\sigma$ have the same dimension, i.e.\ $\HS_\B$ is one-dimensional. Then by the Hadamard Lemma (see \cref{apx:Hadamard}), the target state is
\begin{equation}
e^{-i\rho t} \sigma e^{i\rho t} = \sigma - i[\rho,\sigma] t - \frac{1}{2!}[\rho,[\rho,\sigma]] t^2 + \dotsb.
\end{equation}
We note that for very small evolution times $\Delta$, we have the following direct calculation:
\begin{align}
\Tr_{2} \sof*{e^{-iS\Delta} (\sigma \x \rho) e^{iS\Delta}} &= \sigma - i[\rho,\sigma]\Delta + O(\Delta^2)  \label{eq:lmr1} \\ 
&= e^{-i\rho \Delta} \sigma e^{i\rho \Delta} + O(\Delta^2), \label{eq:lmr2}
\end{align}
were $S$ is the swap operator between the two registers. If we take $\Delta = \delta/t$ and repeat this procedure $O(t^2/\delta)$ times, we end up implementing the operator $e^{-i\rho t}$ up to error $O(\Delta^2 \cdot t^2/\delta) = O(\delta)$. 

Thus the LMR protocol uses $O(t^2/\delta)$ copies of $\rho$ to implement the unitary $e^{-i\rho t}$ up to error $\delta$ in trace norm. To obtain the result for diamond norm, we can simply replace $\rho$ by $\rho_{\A} \otimes \1_{\B}$ and $\sigma$ by $\sigma_{\A\B}$. See \cref{apx:LMR} for a more detailed proof.
\end{proof}

\begin{remark}
While not noted explicitly in \cite{LMR14}, it turns out that the LMR protocol can be implemented efficiently, i.e.\ using $O(\log D \cdot t^2/\delta)$ single-qubit and Fredkin (controlled-swap) gates, where $D = \dim (\HS_\A)$. To see this, we recall that in the proof of Theorem \ref{thm:LMR}, the only potentially expensive operation is the partial swap
\begin{equation}
e^{-iS \Delta} = (\cos \Delta) \1 - i (\sin \Delta) S.
\end{equation}
This operation is a linear combination of the two unitaries $\1$ and $S$, the latter swapping two registers of dimension $D$ and thus being implementable using $O(\log D)$ gates. By the LCU (linear combination of unitaries) algorithm (see e.g. \cite{berry2015hamiltonian}, or \cite[Theorem~2.4]{kothari14}), $e^{-iS \Delta}$ can be implemented using a constant number of uses of controlled-$S$, elementary single-qubit rotations, and a single-qubit unitary $A$ satisfying
\begin{equation}
A\ket{0} \propto \sqrt{\cos\Delta} \ket{0} + \sqrt{\sin\Delta} \ket{1}.
\end{equation}
Hence the LMR protocol can be implemented with $O(\log D \cdot t^2/\delta)$ single-qubit and Fredkin gates.\footnote{This analysis hides the use of the Solovay-Kitaev theorem to implement the single-qubit rotation $A$; decomposed into elementary gates, the total number of gates required is $O((\log D + \poly\log(t/\delta)) \cdot t^2/\delta)$. The runtime stated in the alternative efficient implementation in \cite[Appendix~C]{ML16} similarly does not consider the cost of implementing arbitrary single-qubit rotations.}

We note that Marvian and Lloyd independently give an alternative efficient implementation for the LMR protocol \cite[Appendix~C]{ML16}; their algorithm has an extra multiplicative factor of $\log (t^2/\delta)$ in the runtime.
\end{remark}

\begin{remark}
We also note that the LMR protocol can be modified to implement the controlled-$e^{-i\rho t}$ operation, which will be important if one wants to implement phase estimation on $e^{-i\rho t}$. This fact does not seem to have been addressed in previous work; see \cref{sec:control-LMR}.
\end{remark}

 An alternative method to LMR for the sample-based Hamiltonian simulation problem would be to perform tomography on the
copies of $\rho$ to get an estimate $\widehat{\rho}$ of $\rho$, and then
implement $e^{-i\widehat{\rho} t}$. Let $\zeta=\widehat{\rho}-\rho$ and suppose $\|\zeta\|_1=\epsilon$. We first show that simulating with $\widehat{\rho}$ instead of $\rho$
results in a diamond norm distance at most $\epsilon t$. That is,
\begin{align}\label{eq:diamond_rho}
\left\|e^{-i\widehat{\rho}t}-e^{-i\rho t}\right\|_\diamond \leq\epsilon t.
\end{align}

To show this, we recall the Lie product formula~\cite{Lie}
\begin{align}\label{eq:Lie_1}
e^{-i \widehat{\rho} t} = \lim_{l\rightarrow\infty}
\left(e^{-i\rho t/l}e^{-i\zeta t/l}\right)^l.
\end{align}
For any integer $l \geq 1$ and unitaries $U$ and $V$, using the triangle inequality we have that
\begin{align}
  \norm{U^l - V^l}_\diamond
 &\leq \norm{U^l - U V^{l-1}}_\diamond + \norm{U V^{l-1} - V^l}_\diamond \\
 &   = \norm{U^{l-1} - V^{l-1}}_\diamond + \norm{U - V}_\diamond
\end{align}
where the equality comes from the unitary invariance of the diamond norm. Applying this argument inductively, $\norm{U^l - V^l}_\diamond \leq l \norm{U - V}_\diamond$. Hence
\begin{align}
\norm[\big]{\of[\big]{e^{-i\rho t/l}e^{-i\zeta t/l}}^l-e^{-i\rho t}}_\diamond
&\leq l \norm[\big]{e^{-i\rho t/l}e^{-i\zeta t/l}-e^{-i\rho t/l}}_\diamond \\
&   = l \norm{e^{-i\zeta t/l}-\1}_\diamond \\
&\leq \epsilon t+O\left(\epsilon^2t^2/l\right), \label{eq:diamond_tri}
\end{align}
where the last line follows using a Taylor series expansion.
Finally, using \cref{eq:Lie_1} and taking $l\rightarrow\infty$ in \cref{eq:diamond_tri}, we obtain \cref{eq:diamond_rho}.

Moreover, there exists $\rho$ and $\widehat{\rho}$ for which \cref{eq:diamond_rho} is essentially tight. To see this, note that if $\rho$ and $\widehat{\rho}$ commute, we have
\begin{align}
\left\|e^{-i\widehat{\rho}t}-e^{-i\rho t}\right\|_\diamond &= \left\|\1-e^{-i(\rho-\widehat{\rho}) t}\right\|_\diamond  \\
&= \|\widehat{\rho} - \rho\|_1 t + O(\|\widehat{\rho}- \rho\|_1^2 t^2)\\
&=\epsilon t+O\left(\epsilon^2t^2\right).
\end{align}
This means that if we want to simulate $e^{-i\rho t}$ to error $\delta$ 
in diamond norm, in general we need an estimate $\widehat{\rho}$ such that 
$\|\widehat{\rho}-\rho\|_1=O(\delta/t).$

In Theorem~1 of \cite{HHJ+15}, they prove that to acquire an estimate of a rank-$r$ state $\rho \in \D{d}$ that differs from the true $\rho$ by at most $\epsilon$ in trace distance with probability at least $1-\eta$
requires $n$ copies of $\rho$, where
\begin{align}\label{eq:tomo_lower_bound}
n=\Omega\left(\frac{Cdr\left(1-\epsilon\right)^2}{\epsilon^2\log(d/r\epsilon)}\right)
\end{align}
with $C$ a function of only $\eta$.
The scaling in $\epsilon$ can be slightly improved. Fixing $\eta$, $d$ and $r$, 
\cref{eq:tomo_lower_bound} scales in $\epsilon$ as $\Omega(1/(\epsilon^2\log(1/\epsilon)))$.
If one could acquire such an estimate of $\rho$, one could violate
the Helstrom bound~\cite{Helstrom} that scales as $\Omega(1/\epsilon^2)$. Therefore, we can combine the two bounds to get
\begin{align}\label{eq:better_tomo_lower_bound}
n=\Omega\left(\frac{Cdr(1-\epsilon)^2}{\epsilon^2\log(d/r\epsilon)}+\frac{1}{\epsilon^2}\right).
\end{align}

Back to our problem of Hamiltonian simulation, we want $\epsilon=\delta/t$
to obtain a simulation correct to accuracy $\delta$. Setting $\epsilon=\delta/t$, we find
that the number of samples needed to obtain a tomographic estimate to the desired
accuracy is 
\begin{align}\label{eq:better_tomo_lower_bound2}
n=\Omega\left(\frac{Cdr(t-\delta)^2}{\delta^2\log(dt/r\delta)}+\frac{t^2}{\delta^2}\right).
\end{align}

On the other hand, using LMR to simulate $e^{-i\rho t}$ to accuracy $\delta$ uses
$n$ copies of $\rho$, where
\begin{align}
n=O(t^2/\delta).
\end{align}
Since LMR does not have any dependence on $d$ or $r$, we immediately see that
for large $d$ or $r$, it does significantly better. Furthermore, even fixing
$d$ and $r$, we see that LMR provides a square root improvement in sample complexity over tomography in terms of $\delta$.

\section{LMR Protocol is Optimal}\label{sec:optim}

Our strategy for proving the optimality of the LMR protocol will be as follows. We first give a lower bound on the sample
complexity of distinguishing two specific states. Next, we assume we have a protocol that
simulates $e^{-i\rho t}$ to trace norm (not diamond norm, see discussion below)
$\delta$ using $f(t,\delta)$ samples of $\rho$ for some function $f$. Then we
show that using such a protocol, one can distinguish these two states. However,
if $f=o(t^2/\delta)$, we would violate our lower bound on state discrimination.

For this bound, we will consider states of the form 
\begin{align}\label{eq:rhox}
  \rho(x) &:= x \proj{0} + (1-x) \proj{1}
  = \frac12\1+\left(x-\frac12\right)Z
\end{align}
for some $x \in [0,1]$.

\begin{lemma} \label{lemm: distinguishing multiple copies} Suppose we are
promised that a state $\rho$ is either $\rho(x)$ or
$\rho(x+\epsilon)$ where $x\in(\eta,1-\eta)$ and $\epsilon<\eta<1/2$. 
Then given fewer
than $C_\eta/\epsilon^2$ copies of $\rho$ for some constant $C_\eta$ that depends only on $\eta$, the probability of
correctly determining whether $\rho$ is $\rho(x)$ or
$\rho(x+\epsilon)$ is
at most 2/3.
\end{lemma}

\begin{proof}
If we have $n$
samples of $\rho$, and we want to determine whether $\rho$ is $\rho(x)$ or
$\rho(x+\epsilon)$, then the maximum probability that we correctly
identify $\rho$ is at most~\cite{Holevo73,Helstrom}
\begin{align}\label{eq:prob}
\frac{1+\frac{1}{2}\|\rho(x)^{\otimes n}-\rho(x+\epsilon)^{\otimes n}\|_1}{2}.
\end{align}

Now $\rho(x)$ and $\rho(x+\epsilon)$ commute, so the trace distance in \cref{eq:prob} becomes
the total variation distance between the eigenvalues of the two states. Since
$\rho(x)$ is a rank-2 state, this variation distance is the same as the
variation distance between  two binomial distributions with $n$ trials and
probabilities $x$ and $x+\epsilon$ respectively. Then if $n<C_\eta/\epsilon^2$,
for a sufficiently small constant $C_\eta$ that depends on $\eta$,
as long as $\epsilon<\eta,$ the total variation distance between
these two binomial distributions is less than $1/3$ (from many sources, e.g.~\cite{AJ06}).
\end{proof}

We now show the main result of this section: the sample complexity of the LMR protocol is optimal.

\begin{theorem}\label{thm:LMRopt}
Let $f(t,\delta)$ be the number of copies of $\rho$ required to implement the unitary $e^{-i\rho t}$
up to error $\delta$ in trace norm. Then as long as $\delta \le 1/6$ and $\delta/t \le 1/(6\pi)$, it holds that $f(t,\delta) = \Theta(t^2/\delta)$.
\end{theorem}
The upper bound holds by the LMR protocol, \cref{thm:LMR}, so we will only prove the lower bound. The fact that the trace norm lower bounds the diamond norm makes a tight lower bound in terms of the trace norm a stronger result than if we had used the diamond norm.
\begin{proof}
Given many copies of an unknown state $\rho$, suppose we want to distinguish between the cases $\rho_1 := \rho(\frac{1}{2})$ and $\rho_2 := \rho(\frac{1}{2}+\epsilon)$, with $0<\epsilon\leq 1/2$, promised $\rho$ is one of the two. One way of doing this is to
consider the single-qubit unitary operator $U(\rho,t):=\exp(-i\rho t)$. Then for
\begin{equation}
  t_\epsilon := \frac{\pi}{2} \cdot \frac{1}{\epsilon}
  \label{eq:te}
\end{equation}
the operators $U(\rho_i,t_\epsilon)$ become orthogonal, namely,
\begin{align}
U(\rho_1,t_\epsilon) &\propto \1, &
U(\rho_2,t_\epsilon) &\propto Z,
\end{align}
where $\propto$ indicates that we have hidden an unimportant phase factor.
Consequently, applying $U(\rho,t)$ to $\ket{+}$ and measuring in the $X$-basis will distinguish $\rho_1$ from $\rho_2$ with certainty.

Thus, we can distinguish between $\rho = \rho_1$ or $\rho = \rho_2$ with probability at least $2/3$ using no more than $f(t_\epsilon,1/3)$ copies of $\rho$ by implementing a map that differs from $U(\rho,t_\epsilon)$ by trace norm $1/3$. However \cref{lemm: distinguishing multiple copies} tells us that $C_\eta / \epsilon^2$ samples of $\rho$ are required if $\epsilon < \eta \leq 1/2$. Therefore
\begin{equation}
f(t_\epsilon,1/3) \ge C_\eta / \epsilon^2 = C t_\epsilon^2, \quad t_\epsilon \ge \pi,
\label{eq: large t case}\end{equation}
using \cref{eq:te}, where $C := 4 C_\eta / \pi^2$ is some positive constant. \Cref{eq: large t case} holds whenever $t_\epsilon \geq \pi$ since $\epsilon \leq 1/2$ and so $t_\epsilon = \tfrac{\pi}{2} \cdot \tfrac{1}{\epsilon} \geq \pi$.

Now suppose instead we have arbitrary $\delta$ and $t$ satisfying $\delta \le 1/6$ and $t/\delta \ge 6\pi$, as assumed in the theorem statement. We note the following inequality for any $t \in \R$ and any integer $m \geq 0$:
\begin{equation}
m f(t,\delta) \ge f(mt,m\delta), \label{eq: segments}
\end{equation}
which holds because one way of simulating $\exp(-i\rho mt)$ up to error $m\delta$ is to run $m$ times a simulation of $\exp(-i\rho t)$ up to error $\delta$. Taking $m= \lceil 1/(6\delta) \rceil$, we have
\begin{align}
f(t,\delta) &\ge f(mt,m\delta)/m  \\
&\ge C(mt)^2/m = Cmt^2 \label{eq:cmt} \\
&= \Omega(t^2/\delta),
\end{align}
where \cref{eq:cmt} holds because $m\delta \le 1/6+\delta \le 1/3$ and $mt \ge t/(6\delta) \ge \pi$, so \cref{eq: large t case} applies.
\end{proof}

\section{Pure State Discrimination and the Optimality of LMR for Pure States} \label{sec:optimal}

In the previous section, we saw that the sample complexity of the LMR protocol cannot in general be improved. However, the specific case of state discrimination that we used in the proof involved mixed states. One is left to wonder whether simulating $\exp({-i\proj{\psi}t})$ for a pure state $\ket{\psi}$ might possibly be more efficient. This relates to a practically relevant question; as we will see in \cref{sec:universal}, the LMR protocol and certain pure states as resources create a universal model for quantum computation.

In this section, however, we show that LMR is also optimal for pure states, at least in the $\delta$ error parameter. What about the $t$ parameter? We argue that we cannot expect to prove a meaningful lower bound on the $t$ dependence in pure state LMR. The reason is that, given any state $\rho$ and promised that $\exp({-i\rho t})$ is periodic (i.e.\ $\exp({-i\rho t_1})=\exp({-i\rho t_2})$ for any $t_2=t_1+kT$ for integer $k$ and real number $T$), we can always simulate the Hamiltonian $\rho$ for an equivalent time $t'\in[0,T)$ instead. Notice first that LMR gives an algorithm for this simulation that takes a finite number of samples for any time $t'\in[0,T)$ and fixed $\delta$. Since we have such an upper bound, any lower bound correct in the $\delta$ scaling but not necessarily in the time scaling will differ by at most a constant from this upper bound. Such a lower bound is not meaningful with respect to any asymptotic scaling. Most relevant to our immediate purpose, this argument applies to pure states: knowing a state is pure, we immediately know its period, namely $2\pi$.

To prove that the LMR protocol is optimal for pure states, we first show that \emph{pure} state discrimination reduces to a problem we call sample-based Grover's search. Then, with the help of known bounds on state discrimination, we prove a lower bound on the efficiency of sample-based Grover's search. Finally, we show that LMR can be used to implement sample-based Grover's search, and therefore find that LMR is optimal in terms of the precision $\delta$.

\subsection{Metrological View of Grover's Search} \label{met_Grover}

While Grover's search \cite{G96} is a well-known quantum mechanical task, it is not often stated in its form as a decision problem, and very rarely \cite{DM15} as a metrological decision problem, where the inputs are unitaries and the output depends on a property that those unitaries either possess or do not possess. This guise is useful for our purposes, however, because the LMR protocol, \cref{eq:map}, allows us to turn metrology problems on states into metrology problems on quantum operations.

\newcommand{\T}{\mathcal{T}}

In the metrological view, Grover's search, or perhaps more precisely amplitude amplification \cite{BHMT02}, is the following problem of parameter estimation. Let $\T$ be a subspace of $\C^{2^q}$. We call $\T$ the \emph{target subspace}. Let $U_\T$ be a unitary acting on $q+1$ qubits such that
\begin{equation}
  U_\T \ket{\phi} \ket{0} =
  \begin{cases}
    \ket{\phi}\ket{1},&\textrm{ if }\ket{\phi}\in\T,\\
    \ket{\phi}\ket{0},&\textrm{ if }\ket{\phi}\perp\T.
  \end{cases}
  \label{eq:UT}
\end{equation}
In this problem, and in the following variations, we will assume access to $U_\T$ and $U_\T\ct$ are free.
For an $q$-qubit unitary $V$, define
\begin{equation}
  \lambda
  = \left|\of{\1 \x \bra{1}} U_\T \of[\big]{\of{V \ket{0}\xp{q}} \x \ket{0}}\right|^2.
\end{equation}
Then in Grover's search, the task is to decide whether $\lambda\ge w$ (for $w>0$) or $\lambda=0$, while using $V$ and $V\ct$ as few times as possible. 
In words, if we call $\ket{s}=V\ket{0}\xp{q}$ the start state, we would like to determine whether the start state has substantial probability mass in the target subspace or none, promised one is the case. If we solve this problem using Grover's search and count the number of uses of $V$ and $V\ct$ required to succeed with probability $1-\epsilon$, we get the standard complexity $\Theta(\log(1/\epsilon)/\sqrt{w})$ \cite{BBHT98,BCDZ99}.

One simple modification of metrological Grover's search is to replace the circuit description of the start state preparation operator $V$ with copies of the start state $\ket{s}$ instead. The problem is now to determine whether
$\lambda = |\of{\1 \x \bra{1}} U_\T \of{\ket{s} \x \ket{0}}|^2$
is at least $w>0$ or equal to zero, promised one is the case, given copies of $\ket{s}$ and unlimited access to $U_\T$ and $U_{\T}^\dagger$. We call this \emph{sample-based Grover's search}. But how many copies of $\ket{s}$ are needed? We will see in \cref{sec:pure_optim} that the answer is $\Theta\of*{\log(1/\epsilon)/w}$, and so we find we have lost the square-root advantage of Grover's search. 

A second variant of metrological Grover's search is to replace both $V$ and $U_\T$ with quantum states. In this form, the problem becomes: given copies of $q$-qubit states $\ket{s}$ and $\ket{t}$, determine whether $\lambda=\abs{\braket{s}{t}}^2$ is at least $w>0$ or equal to zero, promised one is the case. We call this variant \emph{orthogonality testing}. The number of copies of $|s\rangle$ and $|t\rangle$ needed will also turn out to be $\Theta\of*{\log(1/\epsilon)/w}$; see \cref{sec:opt_comm_sim}.

Orthogonality testing is similar to \emph{equality testing}, the problem of deciding whether $\abs{\braket{s}{t}}^2 = 1$ or $\abs{\braket{s}{t}}^2<1-w$ for some $w>0$, for which there is already an optimal (up to log factors) algorithm \cite{BCWD01}.

\subsection{The LMR Protocol is Optimal for Pure States}\label{sec:pure_optim}

To show LMR is optimal for pure states, we begin by showing a lower bound on sample-based Grover's search. Then, we show that sample-based Grover's search can be implemented optimally by LMR.

\begin{lemma}\label{state-based bound}
Sample-based Grover's search with success probability $1-\epsilon$ uses $\Theta\left(\log(1/\epsilon)/w\right)$ copies of $|s\rangle$.
\end{lemma}

\begin{proof}
We will first prove the lower bound. Consider the pure state discrimination problem of Helstrom \cite{Helstrom}. You are given a quantum state $\ket{\psi}$ (of arbitrary dimension) which is either $\ket{\phi_1}$ or $\ket{\phi_2}$, each with probability $1/2$. You are provided with classical descriptions of the two states $\ket{\phi_1}$ and $\ket{\phi_2}$, and asked to decide whether $\ket{\psi} = \ket{\phi_1}$ or $\ket{\psi} = \ket{\phi_2}$. Over all possible measurements one could perform on $\ket{\psi}$, what is the minimum failure rate $\epsilon$ that can be achieved?

Helstrom's bound \cite{Helstrom} states that for any discrimination procedure,
\begin{equation}\label{helstrom_bound}
\epsilon\ge\frac12\left(1-\sqrt{1-|\langle\phi_1|\phi_2\rangle|^2}\right).
\end{equation}
A special case of the bound corresponds to discrimination given $n$ copies of $|\psi\rangle$ instead of one. Then, the problem is to discriminate between $|\phi_1\rangle^{\otimes n}$ and $|\phi_2\rangle^{\otimes n}$, and Helstrom's bound can be rearranged to give
\begin{equation}
n\ge\frac{\log4\epsilon(1-\epsilon)}{\log|\langle\phi_1|\phi_2\rangle|^2}\ge
-\frac{\log(1/4\epsilon)}{\log|\langle\phi_1|\phi_2\rangle|^2}.
\end{equation}
For $|\langle\phi_1|\phi_2\rangle|^2$ close to 1, this is asymptotically the bound
\begin{equation}\label{helstrom_lower}
n=\Omega\left(\frac{\log(1/\epsilon)}{1-|\langle\phi_1|\phi_2\rangle|^2}\right).
\end{equation}

The next step in proving the lower bound is to show that pure state discrimination can be done with sample-based Grover's search as described in \cref{met_Grover}. A similar reduction to state discrimination and Helstrom's bound is the key step in \cite{BCWD01} for proving a lower bound on equality testing. 

Now, recall that sample-based Grover's search requires copies of a $q$-qubit input state $\ket{s}$ and a unitary $U_\T$ that defines a target space $\T$, as in \cref{eq:UT}. We set $\ket{s}$ to the mystery state $\ket{\psi}$.
To choose $U_\T$, we note that for some $w \in [0,1]$ we can write $\ket{\phi_2} = \sqrt{1-w} \ket{\phi_1} + \sqrt{w} \ket{\phi_1^\perp}$ where $\ket{\phi_1^\perp}$ is a normalized state such that $\braket{\phi_1}{\phi_1^\perp} = 0$. Note that $\ket{\phi_1^\perp}$ and $w = \abs{\braket{\phi_1^\perp}{\phi_2}}^2$ are known because we have classical descriptions of $\ket{\phi_1}$ and $\ket{\phi_2}$. We can therefore choose $U_\T$, as defined by \cref{eq:UT}, to be a $\ket{\phi_1^\perp}$-tester by letting $\T = \spn \set{\ket{\phi_1^\perp}}$.

Now, provided enough copies of $\ket{s} = \ket{\psi}$, a sample-based Grover's algorithm with this choice of $U_\T$ will be able to distinguish between $\ket{\psi} = \ket{\phi_1}$ (the $\lambda = 0$ case) and $\ket{\psi} = \ket{\phi_2}$ (the $\lambda \ge w$ case). Since it solves the pure state discrimination problem for states $\ket{\phi_1}$ and $\ket{\phi_2}$, using \cref{helstrom_lower} with $\abs{\braket{\phi_1}{\phi_2}}^2 = 1 - w$ gives us the desired bound. 

Now we prove the upper bound. Given an oracle marking the target space, we should, by just applying $U_\T$ to $\ket{s} \ket{0}$ and measuring the ancilla bit, on average observe at least one positive event after $O(1/w)$ trials if $|s\rangle$ does have some overlap $\ge w$ with the target space. To boost the probability of success from a constant to $\epsilon$ requires only a factor of $\log(1/\epsilon)$ more attempts. This is, of course, the pure state discrimination analogue of a classical randomized algorithm for unstructured search.
\end{proof}

While the classical search algorithm described in the proof is an obvious optimal procedure in sampling complexity, we can also solve sample-based Grover's search with LMR. The ultimate reason to do this is not to give a useful algorithm for sample-based search, but rather to show that LMR is optimal in the number of copies of a pure state $\rho$ that it uses to simulate $e^{-i\rho t}$.

\begin{theorem}\label{LMR_opt}
The number of copies of an unknown pure state $\rho$ required for any algorithm to simulate $e^{-i\rho t}$ to trace norm $\delta$ is $\Omega(1/\delta)$.
\end{theorem}
\begin{proof}
For $\rho=\proj{s}$ a pure state we have
\begin{equation}
  e^{-i\rho t}=e^{-it\proj{s}} = \1 - (1 - e^{-it}) \proj{s}.
\end{equation}
Setting $t = \pi$ this is $e^{-i\pi\proj{s}} = \1 - 2 \proj{s} = R_{\proj{s}}$, the reflection about the start state $\ket{s}$ that plays a crucial role in Grover's algorithm \cite{G96}.

Let us implement sample-based Grover's search using this observation. Since we have unlimited access to $U_{\T}$ and $U_{\T}^\dagger$, we can implement the reflection about the target space as $R_{\T} = U_\T\ct (\1 \x Z)U_\T$ without any sampling. Grover's search finds a state $|t_0\rangle$ in the target space $\T = \set{\ket{t} : U_\T \ket{t} \ket{0} = \ket{t} \ket{1}}$ by first applying $G=-R_{\proj{s}}R_{\T}$ to $|s\rangle$ $O(1/\sqrt{\lambda})$ times, where the initial probability mass of $|s\rangle$ inside the target space is $\lambda$. Next, we apply $U_{\T}$ to the resulting state and a $|0\rangle$ ancilla and then measure the ancilla to determine whether a state $|t_0\rangle$ within $\mathcal{T}$ has been found. This process as stated requires knowing $\lambda$ to determine the number of times $G$ is to be applied.

But notice, for the sample-based Grover search problem we are promised only that $\lambda=0$ or $\lambda\ge w$, not that we know $\lambda$ exactly. However, this is not a problem. It has been shown, using either exponentially increasing sequences of Grover iterates \cite{BBHT98} or fixed-point quantum search \cite{YLC14}, that performing $O(\log(1/\epsilon)/\sqrt{w})$ iterates suffices to distinguish the two cases with success probability $1-\epsilon$. That is, in the $\lambda=0$ case, no target state $|t_0\rangle$ is found, while in the $\lambda\ge w$ case such a state is found with probability $1-\epsilon$. In the remainder of the proof, we will take $\epsilon$ to be a constant.

Now since $O(1/\sqrt{w})$ applications of $R_{\proj{s}}$ are required, and we would
like the entire algorithm to succeed with constant error, we need each
application of $R_{\proj{s}}$ to have at most $\sqrt{w}$ error. So if $n_\delta$
copies of $\rho$ are required to simulate $R_{\proj{s}}$ to trace norm accuracy $\delta$,
then $n_{\sqrt{w}} \cdot 1/\sqrt{w}$ copies are required in total for the entire algorithm. Notice $n_{\sqrt{w}}\cdot 1/\sqrt{w}=\Omega(1/w)$ by the lower bound in \cref{state-based
bound}. Thus $n_\delta=\Omega(1/\delta)$
as advertised.
\end{proof}

\section{Generalized LMR for Simulation of Hermitian Polynomials}\label{comm_sim}

The sample-based Hamiltonian simulation of \cref{eq:map} can be further generalized. Instead of evolution of $\sigma$ by a single state $\rho$, the target Hamiltonian $H$ could be encoded by some combination of multiple states $\rho_1,\rho_2,\dotsc,\rho_K$. For example, we might want to implement the map
\begin{equation}
  \sigma \x \bigotimes_{j=1}^K \rho_j^{\otimes n_j} \quad\mapsto\quad
  e^{-i f(\rho_1,\rho_2,\dotsc,\rho_K) t} \sigma e^{i f(\rho_1,\rho_2,\dotsc,\rho_K) t},
  \label{eq:mapf}
\end{equation}
where $H=f(\rho_1,\rho_2,\dotsc,\rho_K)$ is some Hermitian multinomial function of the input states. We will treat this problem fully in this section.

One key tool will be the following lemma, which lets us simulate a Hamiltonian given by the difference of two subnormalized states:
\begin{lemma}
\label{lem:DifferenceOfStates}
Let $\rho' \in \D{\mathbb{C}^2 \otimes \HS_\A}$ be a quantum state of the form $\rho' = \proj{0} \otimes \rho_+ + \proj{1} \otimes \rho_-$, where $\rho_+, \rho_-$ are unknown subnormalized states with $\Tr \rho_+ + \Tr \rho_- = 1$. Using $n$ samples of $\rho'$, a quantum algorithm can transform $\sigma_{\A\B}$ into $\tilde{\sigma}_{\A\B}$ such that
\begin{equation}
  \frac{1}{2}\norm[\big]{
    \of[\big]{e^{-i H t} \x \1_\B} \sigma_{\A\B} \of[\big]{e^{i H t} \x \1_\B}
    -\tilde{\sigma}_{\A\B}
  }_{1} \leq O(\delta), \quad H=\rho_+ - \rho_-,
\end{equation}
if $n = O(t^2/\delta)$.
\end{lemma}
\begin{proof}
We will apply a modification of the LMR protocol---instead of using partial swaps, we will use the following unitary operator:
\begin{equation}
  e^{-iS'\Delta} = \proj{0} \otimes e^{-iS\Delta} + \proj{1} \otimes e^{iS\Delta}
\end{equation}
where
\begin{equation}
  S' := \proj{0} \otimes S + \proj{1} \otimes (-S)
\end{equation}
and $S$ is the usual swap operator. The first qubit in $e^{-iS'\Delta}$ essentially controls whether the partial swap on the remaining two systems is applied in the forwards or backwards direction in time. Applying $e^{-iS'\Delta}$ to $\rho' \otimes \sigma$, we obtain
\begin{align}
 &e^{-iS'\Delta} \of[\big]{\proj{0} \x \rho_+ \x \sigma + \proj{1} \x \rho_- \x \sigma} e^{iS'\Delta} \\
 &= \proj{0} \x e^{-iS\Delta} (\rho_+ \x \sigma) e^{iS\Delta} + \proj{1} \x e^{iS\Delta} (\rho_- \x \sigma) e^{-iS\Delta}.
\end{align}
Tracing out the second register (the register originally containing $\rho_\pm$), we obtain the state
\begin{equation}
\proj{0} \otimes \of[\big]{\Tr(\rho_+) \sigma - i[\rho_+,\sigma]\Delta} +
\proj{1} \otimes \of[\big]{\Tr(\rho_-) \sigma + i[\rho_-,\sigma]\Delta} + O(\Delta^2).
\end{equation}
Now tracing out the first qubit, we obtain the state
\begin{align}
(\Tr \rho_+ + \Tr \rho_-) \sigma - i[\rho_+-\rho_-,\sigma]\Delta + O(\Delta^2) &= \sigma - i[H,\sigma]\Delta + O(\Delta^2) \\
&= e^{-iH\Delta} \sigma e^{iH\Delta} + O(\Delta^2),
\end{align}
where we've used $\Tr \rho_+ + \Tr \rho_- = 1$ and $H = \rho_+ - \rho_-$ in the first line. Thus with one copy of $\rho'$ we can simulate the operation $e^{-iH\Delta}$ up to error $O(\Delta^2)$; by choosing $\Delta = \delta/t$ and repeating this for $t^2/\delta$ times, we obtain a simulation of $e^{-iHt}$ up to error $O(\delta)$, using $O(t^2/\delta)$ copies of $\rho'$.
\end{proof}

\subsection{Simulating Linear Combinations}

In the simplest case where $H=\sum_{j=1}^K c_j\rho_j$ is a linear combination of the $\rho_j$, we show:
\begin{theorem}
\label{cor:LinearCombinationOfStates}
Let $\rho_1,\dotsc,\rho_K \in \D{\HS_\A}$ and $\sigma_{\A\B} \in \D{\HS_\A \x \HS_\B}$ be unknown quantum states, and let $c_1,\dotsc,c_K \in \R$. Using $n$ samples from the states $\{\rho_1,\dotsc,\rho_K\},$ a quantum algorithm can transform $\sigma_{\A\B}$ into $\tilde{\sigma}_{\A\B}$ such that
\begin{equation}
  \frac{1}{2}\norm[\big]{
    \of[\big]{e^{-i H t} \x \1_\B} \sigma_{\A\B} \of[\big]{e^{i H t} \x \1_\B}
    -\tilde{\sigma}_{\A\B}
  }_{1} \leq O(\delta), \quad H=\sum_{j=1}^K c_j\rho_j,
\end{equation}
if $n=O(c^2t^2/\delta)$ where $c := \sum_{j=1}^K|c_j|$. Moreover, on average, the number of copies of $\rho_j$ consumed is $n_j=O(|c_j|ct^2/\delta)$.
\end{theorem}

\begin{proof}
Define
\begin{equation}
\rho' := \frac{1}{c} \sof[\bigg]{ \proj{0} \otimes \sum_{j:c_j>0} c_j\rho_j + \proj{1} \otimes \sum_{j:c_j<0} (-c_j)\rho_j }.
\end{equation}
Note that $\rho'$ is a valid density matrix, and can be created by sampling the
state $\proj{0} \otimes \rho_j$ with probability $c_j/c$ if $c_j > 0$,
and otherwise sampling the state $\proj{1} \otimes \rho_j$ with
probability $-c_j/c$ if $c_j < 0$. (This works for the same reason a mixed state is indistinguishable from the corresponding probabilistic distribution of pure states.) By \cref{lem:DifferenceOfStates}, since $\rho'$ is of the form $\proj{0} \otimes \rho_+ + \proj{1} \otimes \rho_-$ with
\begin{equation}
\rho_+ - \rho_- = \frac{1}{c} \sum_{j=1}^K c_j\rho_j = \frac{H}{c},
\end{equation}
 we can simulate $e^{-iHt} = e^{-i(H/c) (ct)}$ to error $\delta$ using $O(c^2t^2/\delta)$ copies of $\rho'$. On average the state $\rho_j$ is sampled $|c_j|/c \cdot O(c^2t^2/\delta) = O(|c_j|ct^2/\delta)$ times.
\end{proof}

We now show that \cref{cor:LinearCombinationOfStates} is tight. 

\begin{theorem}\label{thm:lin_sim_opt}
Let $\{c_1,\dotsc,c_K\}$ be a set of $K$ real numbers. Then there exist $\rho_1,\dotsc,\rho_K$ such that to simulate $H=\sum_{j=1}^K c_j\rho_j$ for time $t$ and to error $\delta$  in trace norm 
requires $\Omega(c^2t^2/\delta)$ copies of states in $\set{\rho_1,\dotsc,\rho_K}$, where $c:=\sum_j\abs{c_j}$, as long as $\delta$ and $\delta/(ct)$ are smaller than some constants.
\end{theorem}

\begin{proof}
We first consider the case that $\rho_j=\rho$ and $c_j\geq 0$ for all $j$. Then
we can use \cref{cor:LinearCombinationOfStates} to simulate $H=\sum_jc_j\rho = c\rho$ for time
$t$ to accuracy $\delta$ using $O(c^2t^2/\delta)$ samples of $\rho$. Comparing with our lower bound in \cref{thm:LMRopt}, we find
this is optimal.

If we have $c_j$'s such that some $c_j<0$, without loss of generality, assume that $\sum_{j:c_j>0}c_j\ge\sum_{j:c_j<0}|c_j|$. Then if $c_j>0$, set $\rho_j=\rho$, and if $c_j<0$, set $\rho_j$ equal to the maximally mixed state. This gives
\begin{equation}
H = \sum_{j:c_j>0}c_j\rho + c'\1
\end{equation}
where $c'$ is some real number. The term proportional to the identity can be dropped (since it only gives an irrelevant phase factor), and $\sum_{j:c_j>0}c_j \ge c/2$ by assumption, so simulating $H$ for time $t$ to accuracy $\delta$ requires $\Omega(c^2t^2/\delta)$ samples of $\rho$.
\end{proof}

\subsection{Simulating the Commutator and Anticommutator}

In this section, we will show how to simulate a Hamiltonian that is the commutator of two states $H_c=i[\rho_1,\rho_2]$; or the anticommutator $H_a=\{\rho_1,\rho_2\}$; or some linear combination of the two.
As we will show in the next section, simulating the commutator $H_c$ is useful
for orthogonality testing and for adding two unknown pure states. 

One approach to simulating $H_c$ would be to use the expression~\cite[Corollary 2.12.5]{Lie}
\begin{align}
e^{-i H_c t}=\lim_{r\rightarrow \infty}(e^{-i \rho_1 \sqrt{t/r}}e^{-i \rho_2 \sqrt{t/r}}e^{i \rho_1 \sqrt{t/r}}e^{i \rho_2 \sqrt{t/r}})^r
\end{align}
and apply \cref{thm:LMR} sequentially for each term in the product.
Unfortunately, this leads to an error of $O(t)$, and is incorrect even at the lowest order. We now present an alternate approach using \cref{lem:DifferenceOfStates}.

\begin{theorem}
\label{thm:CommutatorSimulationC}
Let $\rho_1,\rho_2 \in \D{\HS_\A}$ and $\sigma_{\A\B} \in \D{\HS_\A \x \HS_\B}$ be unknown quantum states, and $\phi \in [0,2\pi)$.  Using $n$ samples each of $\rho_1$ and $\rho_2$, a quantum algorithm can transform $\sigma_{\A\B}$ into $\tilde{\sigma}_{\A\B}$ such that
\begin{equation}
  \frac{1}{2}\norm[\big]{
    \of[\big]{e^{-i H t} \x \1_\B} \sigma_{\A\B} \of[\big]{e^{i H t} \x \1_\B}
    -\tilde{\sigma}_{\A\B}
  }_\diamond \leq O(\delta),
  \end{equation}
where 
\begin{equation}
H= \frac{1}{2}\cos(\phi)\{\rho_1,\rho_2\} + \frac{1}{2}\sin(\phi)i[\rho_1,\rho_2]
\end{equation}
  if $n=O(t^2/\delta)$.
\end{theorem}
Note that choosing $\phi = 0$ we recover the anticommutator Hamiltonian $H_a/2$, and choosing $\phi = \pi/2$ we recover the commutator Hamiltonian $H_c/2$.

\begin{proof}

For simplicity we only consider the case when $\rho_1$, $\rho_2$, $\sigma\in
\D{\HS_\A}$; the general case can be straightforwardly tackled as in
\cref{apx:LMR}. Our strategy will be to produce a state of the form $\rho'=\ketbra{0}{0}\otimes\rho_++\ketbra{1}{1}\otimes\rho_-$, such that 
\begin{align}
\Tr(\rho_+ + \rho_-) &= 1, &
\rho_+ - \rho_- &= H,
\end{align}
and then apply \cref{lem:DifferenceOfStates}.
We will use the circuit in \cref{fig:SwapPic} to produce such a state. 

\begin{figure}[ht] 
\centering
\includegraphics[width=0.6\textwidth]{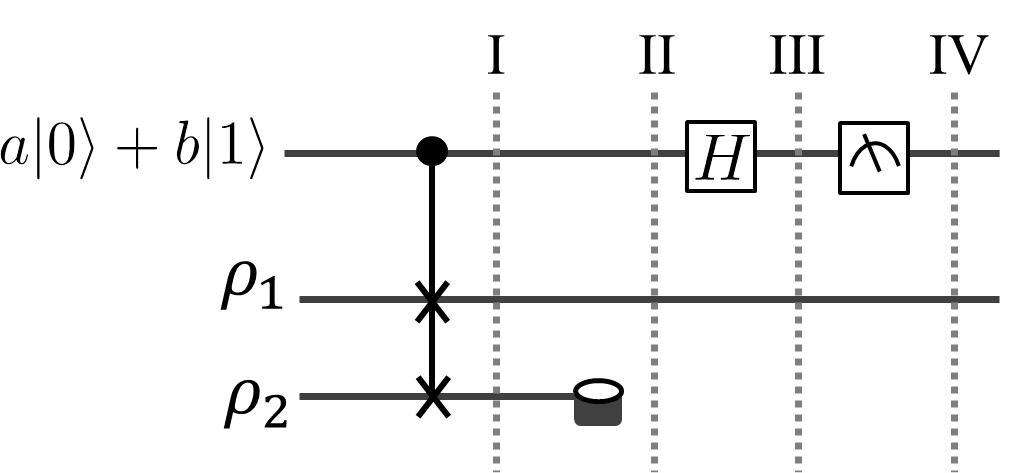}
\caption{The gadget to create a state $\rho'$. Here the controlled-cross-cross gate is a controlled-swap, and the waste bin indicates the partial trace. The $H$-gate is a single-qubit Hadamard gate (not to be confused with the Hamiltonian) and the measurement is in the $Z$-basis.}
\label{fig:SwapPic}
\end{figure}

We now analyze \cref{fig:SwapPic}.
There are a number of dotted lines cutting the figure, and we will write down the
state (as a density matrix) at each. First, at I, after applying a controlled swap,
\begin{align}
  \rho_\textrm{I}
  =&\, \abs{a}^2 \ketbra{0}{0} \x   (\rho_1 \x \rho_2)
       + a b^* \ketbra{0}{1} \x   (\rho_1 \x \rho_2) S \\
      &+ a^* b \ketbra{1}{0} \x S (\rho_1 \x \rho_2)
   + \abs{b}^2 \ketbra{1}{1} \x   (\rho_2 \x \rho_1). \nonumber
\end{align}
After discarding the last register we get
\begin{align}
  \rho_\textrm{II}
  =&\, \ketbra{0}{0} \x \abs{a}^2 \rho_1
   + \ketbra{0}{1} \x a b^*     \rho_1 \rho_2 \label{eq:rho2} \\
  &+ \ketbra{1}{0} \x a^* b     \rho_2 \rho_1
   + \ketbra{1}{1} \x \abs{b}^2 \rho_2.       \nonumber
\end{align}
Finally, a Hadamard operation to the first qubit gives
\begin{align}
\rho_\textrm{III}
=&\,\ketbra{0}{0}\otimes\frac12\left(|a|^2\rho_1+|b|^2\rho_2+ab^*\rho_1\rho_2+a^*b\rho_2\rho_1\right)\\
&+\ketbra{0}{1}\otimes\frac12\left(|a|^2\rho_1-|b|^2\rho_2-ab^*\rho_1\rho_2+a^*b\rho_2\rho_1\right) \nonumber\\
&+\ketbra{1}{0}\otimes\frac12\left(|a|^2\rho_1-|b|^2\rho_2+ab^*\rho_1\rho_2-a^*b\rho_2\rho_1\right) \nonumber\\
&+\ketbra{1}{1}\otimes\frac12\left(|a|^2\rho_1+|b|^2\rho_2-ab^*\rho_1\rho_2-a^*b\rho_2\rho_1\right).\nonumber
\end{align}

Now the measurement operator on the first qubit in \cref{fig:SwapPic} denotes dephasing in the standard basis.
Namely, we measure in the $\{\ketbra{0}{0},\ketbra{1}{1}\}$ basis, and if outcome
$\ketbra{0}{0}$ is obtained, replace the qubit state with $\ketbra{0}{0}$, and if
outcome $\ketbra{1}{1}$ is obtained, replace it with $\ketbra{1}{1}$. 
Thus at IV, after performing this measurement, we have
\begin{align}
\rho_\textrm{IV}=&\,\ketbra{0}{0}\otimes\frac12\left(|a|^2\rho_1+|b|^2\rho_2+ab^*\rho_1\rho_2+a^*b\rho_2\rho_1\right)\\
&+\proj{1}\otimes\frac12\left(|a|^2\rho_1+|b|^2\rho_2-ab^*\rho_1\rho_2-a^*b\rho_2\rho_1\right),\nonumber
\end{align}

Notice that $\rho_\textrm{IV}$ is a state of the form $\rho' = \proj{0} \otimes \rho_+ + \proj{1} \otimes \rho_- $, where $\rho_+$ and $\rho_-$ are subnormalized states with $\Tr \rho_+ + \Tr \rho_- = 1$, and
\begin{equation}
\rho_+-\rho_-=ab^*\rho_1\rho_2+a^*b\rho_2\rho_1.
\end{equation}
Choosing $a=1/\sqrt{2}$ and $b=e^{-i\phi}/\sqrt{2}$ we get
\begin{equation}
  \rho_+ - \rho_- = \frac{1}{2} \of*{e^{i\phi} \rho_1 \rho_2 + e^{-i\phi} \rho_2 \rho_1} = H.
\end{equation}
Applying \cref{lem:DifferenceOfStates}, we can simulate $H$ using the claimed resources. 
\end{proof}

It is easy to see that the simulation from \cref{thm:CommutatorSimulationC} of the anticommutator $H_a=\{\rho_1,\rho_2\}$ has optimal scaling in $t$ and $\delta$, because in the qubit case, we can always choose $\rho_2=\1/2$ so that $H_a=\rho_1$ and we can apply the lower bound from \cref{thm:LMRopt}. It is a little less trivial to show that our simulation of $H_c=i[\rho_1,\rho_2]$ is optimal, but we show now that it is. The proof proceeds along similar lines as the optimality proofs in \cref{thm:LMRopt} and \cref{thm:lin_sim_opt}. 

\begin{theorem}\label{thm:comm_opt}
To simulate $H=i[\rho_1,\rho_2]$ for time $t$ and to trace norm error $\delta$
requires $\Omega(t^2/\delta)$ copies each of the states $\rho_1$ and $\rho_2$, as long as $\delta$ and $\delta/t$ are smaller than some constants.
\end{theorem}

\begin{proof}
First, consider the two states $\rho_A=\rho(1/2)=\1/2$ and $\rho_B=\rho(1/2+\epsilon)$, where $\rho(x)$ is from \cref{eq:rhox} and $0<\epsilon\leq 1/2$. By using the commutator simulation, we will identify an unknown state $\rho_1$ as either $\rho_A$ or $\rho_B$ with probability 2/3, a task for which \cref{lemm: distinguishing multiple copies} gives a lower bound of $C\epsilon^{-2}$ on the sample complexity, for some constant $C$.

Let $\rho_2=\proj{+}$. Then
$\exp({[\rho_1,\rho_2]\pi/(2\epsilon)})\propto \1$ if $\rho_1=\rho_
A$ and
$\exp({[\rho_1,\rho_2]\pi/(2\epsilon)})\propto Y$ if $\rho_1=\rho_B$. A single qubit experiment
then distinguishes $\rho_1=\rho_A$ from $\rho_1=\rho_B$. Simply perform
$\exp({[\rho_1,\rho_2]\pi/(2\epsilon)})$ on $\ket{0}$ through commutator
simulation and measure in the $Z$-basis. The outcome will be $\ket{1}$ if and
only if $\rho_1=\rho_2$. The remaining part of the proof, extending to any
sufficiently large $t$ and small $\delta$, proceeds exactly as in
\cref{thm:LMRopt}. Notice that symmetry of the commutator requires that any
lower bound proved on the number of copies of $\rho_1$ also applies to the
number of copies of $\rho_2$.
\end{proof}

While the above proof uses mixed states, it is possible to prove commutator simuation is optimal for pure states as well, by using the lower bound on orthogonality testing we will provide in \cref{sect:ortho}.

\subsection{Simulating Hermitian Polynomials in the Input States}

It is not hard to show that any Hamiltonian written as a sum of nested commutators (with factors of $i$) and anticommutators can be expanded into a Hermitian multinomial. In fact, the converse is true too, as we sketch in \cref{sect:Jordan}. This motivates us to extend the ideas of \cref{thm:CommutatorSimulationC} to simulate any Hermitian multinomial in the states $\rho_1,\dotsc,\rho_K$, given sample access to these states.

\begin{theorem}\label{thm:LieAlgebra}
Let $\rho_1,\dotsc,\rho_K \in \D{\HS_\A}$ and $\sigma_{\A\B} \in \D{\HS_\A \x \HS_\B}$ be unknown quantum states, and let
\begin{equation}
H = \sum_{r\in R} c_rH_r, \quad H_r = \frac12\left(e^{i\phi_r}\rho_{r_1}\rho_{r_2}\dotsm\rho_{r_{|r|}} + e^{-i\phi_r}\rho_{r_{|r|}}\rho_{r_{|r|-1}}\dotsm\rho_{r_1}\right)
\end{equation}
be a Hermitian polynomial in $\rho_1,\dotsc,\rho_K$, where $R$ is a finite set of strings over the alphabet $\{1,2,\dotsc,K\}$. Using $n$ samples from the states $\{\rho_1,\dotsc,\rho_K\}$, a quantum algorithm can transform $\sigma_{\A\B}$ into $\tilde{\sigma}_{\A\B}$ such that
\begin{equation}
  \frac{1}{2}\norm[\big]{
    \of[\big]{e^{-i H t} \x \1_\B} \sigma_{\A\B} \of[\big]{e^{i H t} \x \1_\B}
    -\tilde{\sigma}_{\A\B}
  }_{1} \leq O(\delta),
\end{equation}
if $n=O(Lc^2t^2/\delta)$ where $c := \sum_{r \in R} \abs{c_r}$ and $L:=\max_{r\in R}|r|$ is the multinomial degree of $H$. Moreover, on average, the number of copies of $\rho_j$ consumed is $n_j=O\left(\kappa_jc^2t^2/\delta\right)$ where $\kappa_j=\sum_{r\in R}v_j(r)|c_r|/c$, and $v_j(r)=|\{s:r_s=j\}|$.
\end{theorem}

\begin{figure}
\centering
\includegraphics[width=0.55\textwidth]{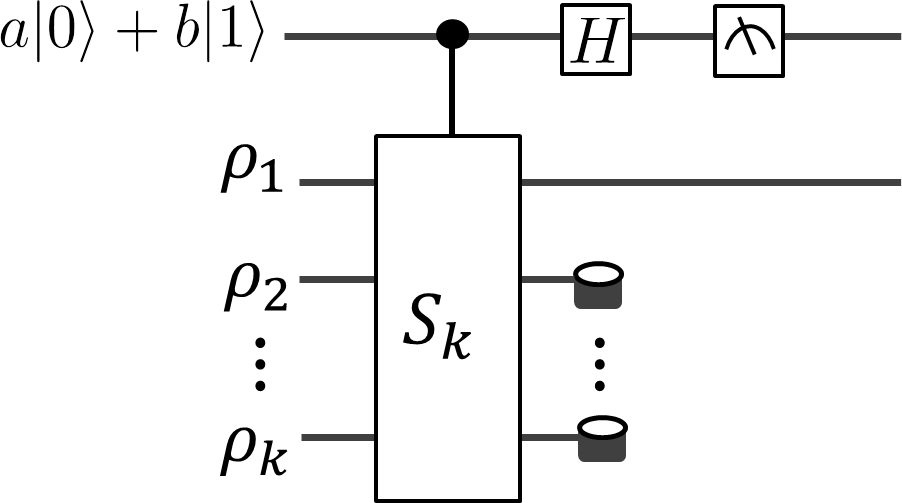}
\caption{The gadget to create $\rho'^{(r)}$. Here $S_k$ is the permutation of $k$ registers given in \cref{eq:perm}, and the waste bins indicate the partial trace. The $H$-gate is a single-qubit Hadamard gate and measurement is in the $Z$-basis. In \cite{Ekert2002} they use the same circuit, but use the measurement outcomes to perform spectrum estimation.}
\label{fig:lie_alg_gadget2}
\end{figure}

\begin{proof}
We first consider a term $H_r$ with $r=(1,2,\dotsc,k)$, for some $k$ such that $2\leq k\leq K$. (More general $r$ will follow easily from this special case.) Let $S_k$ be the cyclic permutation of $k$ copies of $\HS_\A$ that acts as follows: $S_k \ket{j_1, j_2, \dotsc, j_k} = \ket{j_k, j_1, \dotsc, j_{k-1}}$. In other words,
\begin{equation}\label{eq:perm}
S_k:=\sum_{j_1,j_2,\dotsc,j_k=1}^{\dim \HS_\A}\ketbra{j_k}{j_1}\otimes\ketbra{j_1}{j_2}\otimes\ketbra{j_2}{j_3}\otimes\dotsb\otimes\ketbra{j_{k-1}}{j_k}.
\end{equation}

Consider the circuit in \cref{fig:lie_alg_gadget2}. The output of 
\cref{fig:lie_alg_gadget2}  (we will not go through the details of the
calculation as they are very similar to the analysis in
\cref{thm:CommutatorSimulationC}) is of the form
$\rho'^{(r)}=\proj{0}\x\rho_+^{(r)}+\proj{1}\x\rho_-^{(r)}$, where
\begin{align}
\rho_+^{(r)}&:=\frac12\left(|a|^2\rho_1+|b|^2\rho_k+ab^*\rho_1\rho_2\dotsm\rho_k+a^*b\rho_k\rho_{k-1}\dotsm\rho_1\right),\nonumber\\
\rho_-^{(r)}&:=\frac12\left(|a|^2\rho_1+|b|^2\rho_k-ab^*\rho_1\rho_2\dotsm\rho_k-a^*b\rho_k\rho_{k-1}\dotsm\rho_1\right).
\end{align}
When we chose $ab^* = e^{i\phi_r}/2$, we find
\begin{equation} \label{eq: tau_diff_raw}
\rho_+^{(r)}-\rho_-^{(r)}=\frac{1}{2}e^{i\phi_r}\rho_{1}\rho_{2}\dotsm\rho_{{k}}+\frac{1}{2}e^{-i\phi_r}\rho_{{k}}\rho_{{k-1}}\dotsm\rho_{1}=H_r.
\end{equation} 
To apply this analysis to any other $r$ with $|r|=k$, one can simply supply the appropriate input states $\rho_j$ in \cref{fig:lie_alg_gadget2}.

Now without loss of generality we can assume $c_r\ge 0$ for all $r$, since the sign can be absorbed into the phase $\phi_r$. Therefore by sampling from $r\in R$ with probability $c_r/c$ and creating $\rho'^{(r)}$, we obtain the state
\begin{align}
\rho'=\frac1c\left(\sum_{r\in R}c_r\rho'^{(r)}\right)=\frac1c\left(\proj{0}\x\left(\sum_{r\in R}c_r\rho_+^{(r)}\right)+\proj{1}\x\left(\sum_{r\in R}c_r\rho_-^{(r)}\right)\right).
\end{align}
By \cref{lem:DifferenceOfStates}, we can therefore simulate the Hamiltonian
\begin{equation}
H=\sum_{r\in R}c_r(\rho_+^{(r)}-\rho_-^{(r)})=\sum_{r\in R}c_rH_r
\end{equation}
for the desired time and precision using $O(c^2t^2/\delta)$ copies of $\rho'$.
Since each copy of $\rho'$ requires a sample of a state $\rho'^{(r)}$, and each of these states requires at most $L=\max_{r\in R}|r|$
copies of states in $\{\rho_1,\dots,\rho_K\}$, we obtain the stated total
sample complexity.

To calculate the average number of uses of $\rho_j$, we note that $\rho_j$ is used $v_j(r)$ times to create the state $\rho'^{(r)}$, and to create the state $\rho'$, the state $\rho'^{(r)}$ is chosen with probability $|c_j|/c.$ Thus $\rho_j$ is used on average $\kappa_j=\sum_{r\in R}v_j(r)|c_r|/c$ times to create a single $\rho'.$ Then since $O(c^2t^2/\delta)$ copies of $\rho'$ are used in the simulation, we obtain the stated complexity.

\end{proof}

\section{Applications of Commutator Simulation} \label{sec:opt_comm_sim}

One might wonder if commutator simulation is useful for any quantum information processing task. We describe how commutator simulation can be used to coherently add two pure states, i.e. producing a state proportional to $\ket{\psi_1} + \ket{\psi_2}$. We also show that commutator simulation can be used to perform orthogonality testing. Recall from \cref{met_Grover} that this is the problem of determining whether two pure states have overlap at least $w$ or are orthogonal. 

\subsection{Coherent state addition}

We first give a protocol for coherent state addition: given many copies of unknown pure states $\ket{\psi_1}$ and $\ket{\psi_2}$, the task is to obtain a state of the form
\begin{equation}
  a \ket{\psi_1} + b \ket{\psi_2}
  \label{eq:ab}
\end{equation}
for some $a,b \in \R$. Note that the target state is sensitive to the global phases of the two input states---in particular, the relative phase between $\ket{\psi_1}$ and $\ket{\psi_2}$---which have no physical meaning. To make the task well-defined, we instead demand the target state to be of the form
\begin{equation}
  a \ket{\psi_1} + b \frac{\braket{\psi_2}{\psi_1}}{\abs{\braket{\psi_2}{\psi_1}}} \ket{\psi_2}
  \label{eq:ab'}
\end{equation}
for some $a,b \in \R$, which is unique (up to a global phase) even when the global phases of the two input states have not been specified. Note that we can always recover \cref{eq:ab} from \cref{eq:ab'} by fixing the global phases of the two input states appropriately (i.e. such that $\braket{\psi_2}{\psi_1} \geq 0$).

\begin{theorem}\label{thm:addition}
Let $\ket{\psi_1}$ and $\ket{\psi_2}$ be unknown pure states of the same dimension. Promised that the angle between the two states is $\Delta := \arccos \abs{\braket{\psi_1}{\psi_2}}$ and $\Delta \notin \set{0,\pi/2}$, it is possible to create the state
\begin{equation}
  \ket{\psi(\chi)} := \frac{1}{\sin \Delta}
  \of[\Big]{ \sin (\Delta - \chi) \ket{\psi_1} + e^{i\varphi} \sin \chi \ket{\psi_2} }
  \label{eq:chi}
\end{equation}
to trace distance $\delta$ using $O(\frac{\chi^2}{\delta \sin^2 2 \Delta})$ copies of $\ket{\psi_1}$ and $\ket{\psi_2}$, where $e^{i\varphi} := \braket{\psi_2}{\psi_1} / \abs{\braket{\psi_2}{\psi_1}}$ is an unimportant phase factor that can be ignored by appropriately adjusting the global phases of the two states.
\end{theorem}

 \begin{remark}
 Our proof is based on commutator simulation and effectively implements a rotation in the two-dimensional subspace spanned by $\ket{\psi_1}$ and $\ket{\psi_2}$. Indeed, note from \cref{eq:chi} that $\ket{\psi(0)} = \ket{\psi_1}$ and $\ket{\psi(\Delta)} = e^{i \varphi} \ket{\psi_2}$, while intermediate values of $\chi$ produce states that interpolate between these two. As a consequence, the target state in \cref{eq:ab'} has real coefficients $a$ and $b$. One can also achieve complex coefficients using a more sophisticated Hamiltonian that includes terms proportional to $\proj{\psi_1}$ and $\proj{\psi_2}$, but we do not consider this case here for the sake of simplicity. 
 \end{remark}

 \begin{remark}
 If one does not care about the relative phase $e^{i \varphi}$, one can always exchange the two states and replace $\chi$ by $\Delta-\chi$, which would improve the complexity by a constant factor when $\chi > \Delta/2$.
 \end{remark}
 
 \begin{remark}
 Our protocol requires a very large number of samples when the states $\ket{\psi_1}$ and $\ket{\psi_2}$ have either very small or very large overlap (i.e.\ in cases when $\sin^2 2 \Delta$ is very small). This is because we use commutator simulation to effectively implement a rotation in the two-dimensional subspace spanned by $\ket{\psi_1}$ and $\ket{\psi_2}$, and in the special cases when $\ket{\psi_1} \perp \ket{\psi_2}$ or $\ket{\psi_1} = e^{i \varphi} \ket{\psi_2}$ the commutator vanishes and hence our protocol fails (in the second case the task is trivial though).
 \end{remark}

\begin{proof}
Using $\braket{\psi_2}{\psi_1} = e^{i \varphi} \cos \Delta$, we can write
\begin{equation}
  \ket{\psi_2} = e^{-i\varphi} \of[\big]{ \cos \Delta \ket{\psi_1} - \sin \Delta \ket{\psi_1^\perp}}
  \label{eq:psi2}
\end{equation}
for some unit vector $\ket{\psi_1^\perp}$ such that $\braket{\psi_1}{\psi_1^\perp} = 0$.

Then the commutator of two non-orthogonal pure states acts as a Hamiltonian that induces a rotation in the two-dimensional subspace spanned by the states. In particular,
\begin{align}
  i \sof[\big]{\proj{\psi_1},\proj{\psi_2}}
  & = i \of[\big]{\braket{\psi_1}{\psi_2} \ketbra{\psi_1}{\psi_2}
                 -\braket{\psi_2}{\psi_1} \ketbra{\psi_2}{\psi_1}} \\
  & = \cos \Delta \, i
      \of[\big]{e^{-i\varphi} \ketbra{\psi_1}{\psi_2}
               -e^{ i\varphi} \ketbra{\psi_2}{\psi_1}} \\
  & = \cos \Delta \sin \Delta \, i
      \of[\big]{\ketbra{\psi_1^\perp}{\psi_1}
               -\ketbra{\psi_1}{\psi_1^\perp}} \\
  & =: \cos \Delta \sin \Delta \, Y_{\ket{\psi_1},\ket{\psi_1^\perp}},
  \label{eq:Y}
\end{align}
where $Y_{\ket{\psi},\ket{\psi^\perp}}$ acts as the Pauli $Y$ matrix in the two-dimensional subspace spanned by orthonormal states $\ket{\psi}$ and $\ket{\psi^\perp}$. If $Y$ is the $2 \times 2$ Pauli matrix then $e^{i \chi Y} = \cos \chi \, \1 + i \sin \chi \, Y$ for any $\chi \in \R$ so
\begin{align}
  \exp\of[\big]{i \chi Y_{\ket{\psi_1},\ket{\psi_1^\perp}}} \ket{\psi_1}
  &= \cos \chi \ket{\psi_1} + i \sin \chi Y_{\ket{\psi_1},\ket{\psi_1^\perp}} \ket{\psi_1} \\
  &= \cos \chi \ket{\psi_1} - \sin \chi \ket{\psi_1^\perp} \\
  &= \frac{1}{\sin \Delta} \of[\big]{
       \sin (\Delta-\chi) \ket{\psi_1}
     + e^{i\varphi} \sin \chi \ket{\psi_2}
     },
\end{align}
which is the desired state $\ket{\psi(\chi)}$ (we substituted $\ket{\psi_1^\perp} = (\cos \Delta \ket{\psi_1} - e^{i\varphi} \ket{\psi_2}) / \sin \Delta$ from \cref{eq:psi2} to get the last line). To prepare this state, we can apply $\exp\left(i\chi Y_{\ket{\psi_1},\ket{\psi_1^\perp}}\right)$ to $\ket{\psi_1}$ using the commutator simulation algorithm: we evolve with $H := i \sof[\big]{\proj{\psi_2},\proj{\psi_1}} = - \cos \Delta \sin \Delta \, Y_{\ket{\psi_1},\ket{\psi_1^\perp}}$ for time $t = \chi/(\cos \Delta \sin \Delta)$. According to \cref{thm:CommutatorSimulationC}, this requires $O(t^2/\delta)$ copies of each state, so the total sample complexity is $O(\frac{\chi^2}{\delta \sin^2 2\Delta})$.
\end{proof}

Interestingly, by choosing $\chi = \Delta/2$ in \cref{eq:chi} it is possible to \emph{coherently add} two states, i.e.\ create a state proportional to $\ket{\psi_1}+\ket{\psi_2}$ (we are ignoring the relative phase between the two states). However, to determine $\Delta$ one needs to estimate the inner product between the two states, which can be done by running phase estimation on the commutator.

\subsection{Orthogonality Testing}\label{sect:ortho}

We now give a method for testing the orthogonality of two unknown pure states.

\begin{theorem}\label{thm:ortho}
Let $\ket{\psi_1}$ and $\ket{\psi_2}$ be unknown pure states of the same dimension. Promised that either $\abs{\braket{\psi_1}{\psi_2}} = 0$ or $\abs{\braket{\psi_1}{\psi_2}} \geq w$, deciding which with probability $1-\epsilon$ uses $\Theta(\log(1/\epsilon)/w)$ copies of $\ket{\psi_1}$ and $\ket{\psi_2}$.
\end{theorem}

\begin{proof}
For the upper bound, let $\Delta := \arccos \abs{\braket{\psi_1}{\psi_2}}$.
From \cref{eq:Y}, we see that if the states are non-orthogonal, commutator simulation generates a rotation; whereas if the states are orthogonal, i.e.\ $\Delta = \pi/2$, commutator simulation performs only the identity. However, we have to be careful, because identical states correspond to $\Delta = 0$ which also results in a trivial commutator. Therefore we consider the modified states $\ket{\widetilde\psi_1}:=\ket{\psi_1}\ket{0}$ and $\ket{\widetilde\psi_2}:=\ket{\psi_2}\ket{+}$, which can
never have overlap greater than $1/2$. (We do this by appending the states $\ket{0}$ and $\ket{+}$ to the sampled states.) If we let $\lambda := \abs{\braket{\psi_1}{\psi_2}}^2 = \cos^2 \Delta$, then $\abs{\braket{\widetilde\psi_1}{\widetilde\psi_2}}^2 = \lambda/2$ and the evolution with the commutator of $\ket{\widetilde\psi_1}$ and $\ket{\widetilde\psi_2}$ for time $t=1$ generates the unitary
\begin{equation}
  U
  = \exp \of[\Big]{\sof[\big]{\proj{\widetilde\psi_1},\proj{\widetilde\psi_2}}}
  = \exp \of*{-i \sqrt{\frac{\lambda}{2}\of*{1-\frac{\lambda}{2}}}
              Y_{\ket{\widetilde\psi_1},\ket{\widetilde\psi_1^\perp}}}.
\end{equation}
Phase estimation on $U$ to precision $\Omega(1/\sqrt{w})$ suffices to distinguish between $\lambda = 0$ and $\lambda \ge w$, and thus solves orthogonality testing. Similarly to \cref{sec:phase_est}, phase estimation with constant probability of success requires $O(1/\sqrt{w})$ applications of $U$, each implemented to error $O(\sqrt{w})$; this uses $O(1/(\sqrt{w})^2) = O(1/w)$ samples. To succeed with probability $1-\epsilon$ we can repeat $O(\log(1/\epsilon))$ times, giving a total sample complexity of $O(\log(1/\epsilon)/w)$.\footnote{Notice that simple repeating of the SWAP test \cite{BCWD01} on $\ket{\psi_1}$ and $\ket{\psi_2}$ produces a \emph{slower} orthogonality testing algorithm. Essentially, we end up having to distinguish Bernoulli random variables with $p=1/2$ and $p\ge1/2+w/2$. This takes $\Omega(1/w^2)$ samples (see \cref{lemm: distinguishing multiple copies} and \cite{AJ06}).}

For the lower bound, first notice that sample-based Grover search 
(see \cref{sec:pure_optim}) reduces to orthogonality testing in the following way. Since in
sample-based Grover search we do not count uses of $U$, we may therefore perform
tomography on $U$ to learn, to arbitrary accuracy, a complete orthogonal basis
$\{|t_1\rangle,|t_2\rangle,\dots,|t_k\rangle\}$ for the target space $T$. Now
perform orthogonality testing between $|s\rangle$ and each of the $|t_j\rangle$
to determine whether $|s\rangle$ has overlap $\lambda_j=|\langle
s|t_j\rangle|^2$ at least $w/k$. If the total probability mass of $|s\rangle$
inside the target space is at least $w$, then this must be true for some
$|t_j\rangle$. Treating factors of $k$ as constant, this implies
$\Omega(\log(1/\epsilon)/w)$ copies of $|s\rangle$ (and $|t_j\rangle$ by
symmetry) must be required for orthogonality testing with success probability
$1-\epsilon$, so as not to break the sample-based Grover search lower bound of \cref{state-based bound}.
\end{proof}

\section{Universality of LMR} \label{sec:universal}

In many solid state implementations of quantum computers such as quantum dots
\cite{LD98}, donor-pairs \cite{K98}, and electron spins \cite{VYW00}, the
Heisenberg exchange is the natural coupling interaction between qubits. 
Up to an overall scaling, the Heisenberg interaction
is the same as the swap interaction used in the LMR protocol. 

The Heisenberg interaction between
qubits $i$ and $j$ is given by
\begin{align}
H_{ij}=X^i\otimes X^j+Y^i\otimes Y^j+Z^i\otimes Z^j,
\end{align}
where $X^i$, $Y^i$, and $Z^i$ are the Pauli matrices acting on qubit $i.$
In these solid state systems, the Heisenberg interaction can be turned
on and off for different pairs of qubits for any desired length
of time.

The operations induced by the Heisenberg interactions in these systems are fast and
reliable.
While it is beneficial to create computing
models that take advantage of this Heisenberg exchange interaction,
 the Heisenberg interaction is not universal for spin-$1/2$ systems \cite{BBC95}. 
Several schemes have overcome this limitation by using encoded logical qubits and 
decoherence-free-subsystems \cite{Universality,Levy}.

In this section, we show how to use the LMR protocol to design a universal
model for quantum computation that does not use encoded qubits, but which
requires only the Heisenberg interaction, as well as the ability to prepare
the states $\ket{0}$ and $\ket{+}$ on a single qubit. Our scheme thus requires
$n+1$ physical qubits to perform computations on $n$ qubits, in contrast
to encoded schemes, of which the simplest require 2 or 3 times the number
of physical qubits \cite{Universality,Levy}. Furthermore, there has been 
much research in the field of quantum dots on how to quickly and reliably prepare
a fixed qubit state, e.g. \cite{CV10,FPMU03,HVV04,RSL00}. These schemes
could be applied to produce the single qubit states needed for our protocol.

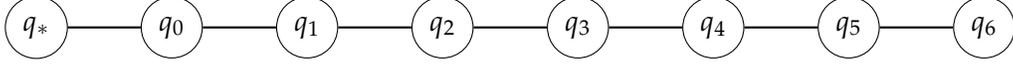
\begin{figure}[ht]
\centering
\begin{tikzpicture}[scale=.9]
\tikzstyle{vertex} = [circle,draw,fill=white,minimum size=.5em]
\tikzstyle{operator} = [rectangle,rounded corners,draw,
fill = black!20!white,
minimum size=1.5em]
\tikzstyle{turn} = [rectangle,draw,fill=white,minimum size=1.5em]
\node[vertex] (v1000) at (0,0) {$q_*$};
\node[vertex] (v1001) at (2,0) {$q_0$};
\node[vertex] (v1010) at (4,0) {$q_1$};
\node[vertex] (v1011) at (6,0) {$q_2$};
\node[vertex] (v1100) at (8,0) {$q_3$};
\node[vertex] (v1101) at (10,0) {$q_4$};
\node[vertex] (v1110) at (12,0) {$q_5$};
\node[vertex] (v1111) at (14,0) {$q_6$};
\path (v1000) edge[thick] (v1001);
\path (v1001) edge[thick] (v1010);
\path (v1010) edge[thick] (v1011);
\path (v1011) edge[thick] (v1100);
\path (v1100) edge[thick] (v1101);
\path (v1101) edge[thick] (v1110);
\path (v1110) edge[thick] (v1111);
\end{tikzpicture}
\caption{Connectivity graph for qubits in our model. Each circle
represents a qubit. Qubits connected by a solid line can have
the Heisenberg interaction applied between them. The qubit
$q_*$ can be prepared in the state $\ket{0}$ or $\ket{+}.$}
\label{fig:connectgraph}
\end{figure}

We consider a connectivity graph of the qubits as in \cref{fig:connectgraph} (different
connectivity graphs lead to different scalings depending on which
costs you would like to optimize). We assume exchange
interactions can be applied between connected qubits in the form
of unitaries $\exp(-i t H_{ij})$ for arbitrary $t$. The qubit
 $q_*$ is where the states $\ket{0}$ and $\ket{+}$ are prepared.

Recall that arbitrary single qubit gates combined with any entangling
two-qubit gate is sufficient for universal quantum computation \cite{BDD02}.
Since we do not have encoded qubits, the exchange interaction itself
immediately gives us an entangling gate.  Now for universal quantum computation
we need to show how to perform arbitrary single qubit gates. 

Let $X_\phi$ denote the unitary operation
\begin{align}
\cos(\phi/2)\1+i\sin(\phi/2)X,
\end{align}
and let $Z_\theta$ denote the unitary operation
\begin{align}
\cos(\phi/2)\1+i\sin(\phi/2)Z.
\end{align}
Then any single qubit rotation can be written as $X_\phi Z_\theta X_\xi$
for angles $\phi$, $\theta$, and $\xi.$ Therefore, it is 
sufficient to show how to perform $X$ and $Z$ rotations.

If qubit $i$ needs to have a single qubit gate performed on it, using
the Heisenberg interaction, we use swap gates to move that qubit to position 0. We now
show how to perform a $Z_{\phi}$ and $X_{\theta}$ on the qubit in
position 0. Using LMR, given $n$ copies of the state $\ket{0}$ input
at qubit $q_*$, using only partial swap operations on qubits $q_0$ and
$q_*$, (i.e.\ applying the Heisenberg interaction between qubits $q_0$ and
$q_*$)  we can apply the unitary
\begin{align}
\exp(-i\phi\proj{0})=Z_{\phi}
\end{align}
to accuracy $O\left(n^{-1}\right)$. Likewise, using the LMR protocol, given $n$ copies of the state $\ket{+}$, using only partial swap interactions between qubits $q_0$ and $q_*$, we can apply the unitary
\begin{align}
\exp(-i\theta\proj{+})=X_{\theta}
\end{align}
to accuracy $O\left(n^{-1}\right)$.

To apply an arbitrary single qubit rotation to accuracy
$\epsilon$, we need $O(\epsilon^{-1})$ resource states $\ket{0}$
and $\ket{+}$ (this construction is reminiscent of ideas in~\cite{MM08}).
Suppose that over the course of an algorithm, one must
apply $M$ single qubit gates and $M'$ CNOT gates. We note that
to apply a CNOT gate requires a constant number of single qubits gates
as well as a constant number of partial swap gates \cite{BDD02}.
Then to bound the error
over the course of the algorithm, we require accuracy of $O((M+M')^{-1})$
for each single qubit gate. Therefore,
we require $O((M+M')^2)$ resource states $\ket{0}$
and $\ket{+}$ in total. 
Additionally, using the connectivity graph of \cref{fig:connectgraph}, to move qubits into proximity with one another to perform any single or two
qubit gate requires $O(N)$ swap operations operations, where $N$ is the number
of qubits. This results in a total number of operations that scales
as $O(N(M+M')^2)$.

We note that the states $\ket{0}$ and $\ket{+}$ need not be prepared perfectly
for our protocol to work. For example, if we have slightly depolarized versions of these
states, we would simply need to increase the number of rounds in the LMR protocol
by a constant factor. In fact, two arbitrary states (other than $\ket{0}$ and $\ket{+}$) could be used, as long as they 
are not diagonal in the same basis, and as long as the states themselves are 
well characterized.

Our model produces a polynomial (in particular squared) blow-up in the number
of operations, which still allows for universal quantum computation. However,
with such a model, it would be impossible to obtain a speed-up for problems
such as Grover's search. We hope it is a useful model for systems where the
Heisenberg exchange is a natural operation. It may even be useful in non-solid
state systems such as cold, trapped atoms, where it was shown that partial
swaps could be implemented using Rydberg interactions or through coupling to a
cavity \cite{PZS+16}.

\section{Outlook}\label{sec:outlook}

We have shown that the LMR protocol is optimal for the problem of simulating unknown Hamiltonians encoded as quantum states. Moreover, the protocol and its generalizations also turn out to be optimal for a variety of other tasks, such as discriminating between pure states and Hamiltonian evolution under the commutators of unknown states. We hope that this study will motivate the discovery of other possible applications of this versatile protocol.

We have not shown the optimality of our protocol for simulating the evolution by the multinomials in \cref{eq:mapf}. It would be interesting to investigate whether it is optimal, or whether better algorithms can be found. 

Another interesting aspect is the role of ancilla qubits in our protocol. While the original LMR protocol for Hamiltonian simulation is based on partial swaps and hence does not require ancilla qubits, the use of ancillas seems to be essential in our more general simulation protocol (see \cref{fig:lie_alg_gadget2}). We wonder whether the use of ancillas is necessary in our protocol, or for example, whether it can be implemented using the continuous permutations introduced in \cite{Ozols15}. These continuous permutations generalize the partial swap operation and do not require ancillas. 

Another possible direction is to investigate distributed versions of our protocols in the context of multiparty communication. \cite{HL11} consider a protocol for simulating distributed unitaries over multiple remote parties using shared entanglement and a limited amount of quantum communication, and the techniques they use are reminiscent to those of the LMR protocol. It would be interesting to investigate connections of \cite{HL11} with the protocols in our work.

Finally, the LMR protocol can be seen as allowing the encoding of the operation $e^{-i \rho t}$ into multiple copies of a quantum state $\rho$. As discussed in \cref{sec:Prelim}, having access to $O(t^2/\delta)$ copies of $\rho$ allows a user to perform the operation $e^{-i \rho t}$, but may be insufficient for the user to determine what $\rho$ is through tomography. It is an intriguing question whether other quantum operations could be encoded into states in this way, so that a user could perform the quantum operation but learn little else about what operation is being performed. This could be seen as a form of quantum copy-protection \cite{Aaronson09}. See \cite{ML16} for some progress in this direction, and \cite{AF16} for negative results when the encoding is required to be a circuit and not a state.

\section*{Acknowledgments}

We thank Andrew Childs for suggesting the proof idea of \cref{LMR_opt}, and Aram Harrow, Stephen Jordan, Seth Lloyd, Iman Marvian, Ronald de Wolf, and Henry Yuen for useful discussions. SK and CYL are funded by the Department of Defense. GHL is funded by the NSF CCR and the ARO quantum computing projects. MO acknowledges Leverhulme Trust Early Career Fellowship (ECF-2015-256) and European Union project QALGO (Grant Agreement No.~600700) for financial support. TJY thanks the DoD, Air Force Office of Scientific Research, National Defense Science and Engineering Graduate (NDSEG) Fellowship, 32 CFR 168a. Part of this work was done while MO was visiting the University of Maryland and MIT, so he would like to thank both institutions for their hospitality.


\bibliographystyle{alphaurl}
\bibliography{References}

\appendix

\section{Proof of Hadamard Lemma} \label{apx:Hadamard}

\begin{lemma}[Hadamard Lemma]\label{lem:Hadamard}
Let $\ad_A (B) := [A,B]$ (also known as the adjoint representation of a Lie algebra). Then
\begin{equation}
  e^{A} B e^{-A}
  = e^{\ad_A} (B)
  = B + [A,B] + \frac{1}{2!} [A,[A,B]] + \dotsb.
\end{equation}
\end{lemma}

\begin{proof}
Let $f(t) := e^{tA} B e^{-tA}$. Since $\frac{d}{dt} e^{tA} = A e^{tA} = e^{tA} A$,
\begin{equation}
  \frac{d}{dt} f(t)
  = e^{tA} A B e^{-tA} - e^{tA} B A e^{-tA}
  = e^{tA} [A,B] e^{-tA}.
\end{equation}
Repeating the same argument inductively, the $n$-th derivative of $f(t)$, for any $n \geq 0$, is
\begin{equation}
  f^{(n)}(t)
 := \frac{d^n}{dt^n} f(t)
  = e^{tA} [A,B]_n e^{-tA}
\end{equation}
where $[A,B]_n := [A,[A,B]_{n-1}]$ and $[A,B]_0 := B$. In particular, note that
\begin{equation}
  f^{(n)}(0) = [A,B]_n = \ad_{A}^n(B).
\end{equation}
The Taylor expansion of $f(t)$ at $t = 0$ then is
\begin{equation}
  f(t)
  = \sum_{n \geq 0} \frac{f^{(n)}(0)}{n!} t^n
  = \sum_{n \geq 0} \frac{[A,B]_n}{n!} t^n
  = \sum_{n \geq 0} \frac{[tA,B]_n}{n!}
  = \sum_{n \geq 0} \frac{\ad_{tA}^n(B)}{n!}
  = e^{\ad_{tA}} (B).
\end{equation}
Recall that $f(t) = e^{tA} B e^{-tA}$, so the result follows by equating the two expression for $f(1)$.
\end{proof}

\section{Proof of the LMR Upper Bound} \label{apx:LMR}

In this section we give a complete proof of \cref{thm:LMR}: there is a protocol that uses $O(t^2/\delta)$ copies of an unknown state $\rho$ to implement the unitary $e^{-i\rho t}$, up to error $\delta$ in diamond norm.
\begin{proof}
Note that $e^{-i \rho_\A \epsilon} \x \1_\B = e^{-i \of{\rho_\A \x \1_\B} \epsilon}$, so we apply the Hadamard Lemma (see \cref{apx:Hadamard}) with $A = -i \of{\rho_\A \x \1_\B} \epsilon$ and $B = \sigma_{\A\B}$. This yields
\begin{equation}
  \of[\big]{e^{-i \rho_\A \epsilon} \x \1_\B} \sigma_{\A\B} \of[\big]{e^{i \rho_\A \epsilon} \x \1_\B}
  = \sigma_{\A\B}
  - i [\rho_\A \x \1_\B, \sigma_{\A\B}] \epsilon
  - \frac{1}{2!} [\rho_\A \x \1_\B, [\rho_\A \x \1_\B, \sigma_{\A\B}]] \epsilon^2
  + \dotsb.
  \label{eq:Taylor}
\end{equation}

We let $\U{\HS}$ be the set of unitary operators in the
Hilbert space $\HS$. Let $S \in \U{\HS_\A \x \HS_{\A_k}}$ be the unitary operator that swaps systems $\A$ and $\A_k$, i.e.\ $S \ket{i}_\A \ket{j}_{\A_k} = \ket{j}_\A \ket{i}_{\A_k}$ for all $i, j \in \set{1, \dotsc, \dim(\HS_\A)}$. Note that $S^2 = \1$ implies $S\ct = S$ so $S \in \Herm{\HS_\A \x \HS_{\A_k}}$. Applying Hamiltonian $S$ for time $\mp \epsilon$ implements the partial swap unitary
\begin{equation}
  e^{\pm i S \epsilon} = \1 \cos \epsilon \pm i S \sin \epsilon.
  \label{eq:exp}
\end{equation}
The LMR algorithm simply applies the swap Hamiltonian $S$ between systems $\A$ and $\A_k$ for some small amount of time $\epsilon$, and then discards $\A_k$ (this is done consecutively for each copy $\rho_{\A_k}$ as $k$ ranges from $1$ to $n$).

The state after the first iteration of the above procedure can be explicitly written as
\begin{align}
 &  \Tr_{\A_1}
    \sof[\big]{
      (e^{-iS_{\A\A_1}\epsilon} \x \1_\B)
      (\sigma_{\A\B} \x \rho_{\A_1})
      (e^{ iS_{\A\A_1}\epsilon} \x \1_\B)
    } \nonumber \\
 &= \sigma_{\A\B} \cos^2 \epsilon
  - i [\rho_\A \x \1_\B, \sigma_{\A\B}] \sin \epsilon \cos \epsilon 
  + \rho_\A \x \Tr_\A (\sigma_{\A\B}) \sin^2 \epsilon \\
 &= \sigma_{\A\B}
  - i [\rho_\A \x \1_\B, \sigma_{\A\B}] \epsilon
  - \of[\big]{\sigma_{\A\B} - \rho_\A \x \Tr_\A (\sigma_{\A\B})} \epsilon^2
  + \dotsb,
    \label{eq:TSS}
\end{align}
where the partial trace can be computed using graphical notation~\cite{WBC15}, and the last line was obtained using the Taylor expansion at $\epsilon = 0$. Note that the difference in trace distance between the ideal state \cref{eq:Taylor} and our first approximation \cref{eq:TSS} is
\begin{equation}
  \frac{1}{2}\norm*{
    \of[\big]{e^{-i \rho_\A \epsilon} \x \1_\B} \sigma_{\A\B} \of[\big]{e^{i \rho_\A \epsilon} \x \1_\B}
  - \Tr_{\A_1}
    \sof[\big]{
      (e^{-iS_{\A\A_1}\epsilon} \x \1_\B)
      (\sigma_{\A\B} \x \rho_{\A_1})
      (e^{ iS_{\A\A_1}\epsilon} \x \1_\B)
    }
  }_1
  \leq O(\epsilon^2).
\end{equation}

\newcommand{\aprx}[1]{\tilde{\sigma}_{\A\B}^{[#1]}}

If we write $\aprx{k}$ to denote the state after $k$ iterations of this procedure (so $\aprx{0} = \sigma_{\A\B}$ denotes the original state and $\aprx{1}$ denotes the state in \cref{eq:TSS}), we get the following recursion from \cref{eq:TSS}:
\begin{equation}
  \aprx{k}
  = \aprx{k-1}
  - i [\rho_\A \x \1_\B, \aprx{k-1}] \epsilon
  - \of[\big]{\aprx{k-1} - \rho_\A \x \Tr_\A (\aprx{k-1})} \epsilon^2
  + O(\epsilon^3).
\end{equation}
By evaluating this recursively, the final state can be expressed as
\begin{align}
  \aprx{n}
 &= \aprx{n-m}
  - i [\rho_\A \x \1_\B, \aprx{n-m}] m \epsilon
  - \of[\big]{\aprx{n-m} - \rho_\A \x \Tr_\A (\aprx{n-m})} m \epsilon^2 \\
 &+ i \sof[\big]{\rho_\A \x \1_\B, i[\rho_\A \x \1_\B, \aprx{n-m}]} \of{1 + 2 + \dotsb + m} \epsilon^2
  + O(\epsilon^3).
  \label{eq:rec}
\end{align}
for any $m \in \set{0, \dotsc, n}$. In particular, for $m = n$ we get
\begin{align}
  \aprx{n}
 &= \sigma_{\A\B}
  - i [\rho_\A \x \1_\B, \sigma_{\A\B}] n \epsilon
  - \of[\big]{\sigma_{\A\B} - \rho_\A \x \Tr_\A (\sigma_{\A\B})} n \epsilon^2 \\
 &+ i \sof[\big]{\rho_\A \x \1_\B, i[\rho_\A \x \1_\B, \sigma_{\A\B}]} \frac{n(n-1)}{2} \epsilon^2
  + O(\epsilon^3).
  \label{eq:last}
\end{align}

Choosing $\epsilon = t/n$ in \cref{eq:last} and comparing this with the desired final state at time $t$ (given by \cref{eq:Taylor}, with $t$ instead of $\epsilon$), we see that
\begin{equation}
  \frac{1}{2}\norm[\big]{
    \of[\big]{e^{-i \rho_\A t} \x \1_\B} \sigma_{\A\B} \of[\big]{e^{i \rho_\A t} \x \1_\B}
  - \aprx{n}
  }_1
  \leq O(n \epsilon^2)
  = O(t^2/n).
\end{equation}
Thus, if we desire a final trace distance accuracy of $\delta$ and
want to implement the Hamiltonian $\rho$ for a time $t$, the LMR protocol uses
$n=O(t^2/\delta)$ copies of $\rho.$
\end{proof}

\section{Controlled Density Matrix Exponentiation} \label{sec:control-LMR}
As we have seen in \cref{thm:LMR}, given an input state $\sigma$ and $O(t^2/\delta)$ copies of another state $\rho$, the LMR protocol allows us to obtain the output state $e^{-i\rho t} \sigma e^{i \rho t}$. In many applications (see e.g. \cite{LMR14,Wang14}) we would like to perform phase estimation on the operator $e^{-i\rho t}$, which requires the ability to apply the controlled-$e^{-i\rho t}$ operation. However it is not immediately obvious to see that the controlled-$e^{-i\rho t}$ operation can be performed with the LMR protocol: since the LMR protocol involves the discarding (tracing out) of quantum registers, it could very well lose any coherence between the $\ket{0}$ and $\ket{1}$ components of the control qubit. Nevertheless, we show this is not the case:

\begin{theorem}\label{thm:control-LMR}
Given $O(t^2/\delta)$ copies of an unknown quantum state $\rho \in \D{\HS_\A}$, the controlled-$e^{-i\rho t}$ operation, $\ketbra{0}{0}~\otimes~\1_\A~+~\ketbra{1}{1}~\otimes e^{-i\rho t}$, can be performed up to error $\delta$ in diamond norm.
\end{theorem}
\begin{proof}
We will give two ways of deriving this result. The first simplest method is to realize that
\begin{equation}
\exp\of[\big]{-i (\ketbra{1}{1} \otimes \rho) t} = \ketbra{0}{0} \otimes \1_\A + \ketbra{1}{1} \otimes e^{-i\rho t}
\end{equation}
and so to simulate the controlled-$e^{-i\rho t}$ operator, we can simply use the $\ketbra{1}{1} \otimes \rho$ as the input states to the LMR protocol instead.

Alternatively, the naive method of replacing the partial swaps in the LMR protocol by controlled versions also works. To see this, let us consider starting from the initial state $\Sigma = (a\ket{0}+b\ket{1})(a^*\bra{0}+b^*\ket{1}) \otimes \sigma \otimes \rho$, applying the controlled-partial swap $\ket{0}\bra{0} \otimes \1 + \ket{1}\bra{1} \otimes e^{-iS\Delta}$. The result is
\begin{align}
&aa^*\ketbra{0}{0}\otimes \sigma \otimes \rho + ba^*\ketbra{1}{0} \otimes (c \1 - isS)(\sigma \otimes \rho) \nonumber \\
&+ ab^*\ketbra{0}{1} \otimes (\sigma \otimes \rho)  (c \1 + isS)+ bb^*\ketbra{1}{1}\otimes (c \1 - isS) (\sigma \otimes \rho) (c \1 + isS)
\end{align}
where we've used the shorthand $s \equiv \sin \Delta$ and $c \equiv \cos \Delta$. By using the identities $\Tr_2[S(\sigma \otimes \rho)] = \rho \sigma$ and $\Tr_2[(\sigma \otimes \rho)S] = \sigma \rho$, we can calculate the resulting state if we trace out the last register: 
\begin{align}
& aa^*\ketbra{0}{0}\otimes \sigma + ba^*\ketbra{1}{0} \otimes (c\sigma - is\rho\sigma) + ab^*\ketbra{0}{1} \otimes (c\sigma+is\sigma\rho)\nonumber\\
&+ bb^*\ketbra{1}{1}\otimes (c^2\sigma-is\rho\sigma+is\sigma\rho+s^2\rho) \\
={}& (\ketbra{0}{0} \otimes \1 +\ketbra{1}{1} \otimes (\1 - i \rho \Delta)) \:\: \Tr_3\Sigma \:\: (\ketbra{0}{0} \otimes \1 + \ketbra{1}{1} \otimes (\1 + i \rho \Delta)) + O(\Delta^2) \\
={}& (\ketbra{0}{0} \otimes \1 +\ketbra{1}{1} \otimes e^{-i\rho \Delta}) \:\: \Tr_3\Sigma \:\: (\ketbra{0}{0} \otimes \1 +\ketbra{1}{1} \otimes e^{-i\rho \Delta}) +O(\Delta^2)
\end{align}
where $\Tr_3\Sigma$ refers to the initial state with the third register ($\rho$) traced out. This shows that we can use one copy of $\rho$ to implement the controlled-$e^{-i\rho \Delta}$ operation up to error $O(\Delta^2)$. Therefore similar to the discussion in the proof of \cref{thm:LMR}, by choosing $\Delta = \delta/t$ and repeating this procedure $O(t^2/\delta)$ times, we implement the controlled-$e^{-i\rho t}$ operator up to error $O(\delta)$.
\end{proof}

\begin{remark}
It is easy to see that the simple trick of appending a qubit in the $|1\rangle$ state, as described in the first paragraph of the proof of \cref{thm:control-LMR}, also works for all results in \cref{comm_sim}. For instance, let $f(\rho_1,\rho_2,\dotsc,\rho_K)$ be an arbitrary Hermitian polynomial in $\rho_1,\dotsc,\rho_K)$. Then $f(\proj{1} \otimes \rho_1,\dotsc,\proj{1} \otimes \rho_K) = \proj{1} \otimes f(\rho_1,\dotsc,\rho_K)$, and hence if we wish to simulate the controlled-$e^{-if(\rho_1,\dotsc,\rho_K)t}$ operator, we have
\begin{align}
\ketbra{0}{0} \otimes \1_\A + \ketbra{1}{1} \otimes e^{-if(\rho_1,\dotsc,\rho_K) t} &= \exp\of[\big]{-i (\ketbra{1}{1} \otimes f(\rho_1,\dotsc,\rho_K)) t} \\
&= \exp\of[\big]{-i f(\ketbra{1}{1} \otimes \rho_1,\dotsc,\ketbra{1}{1} \otimes \rho_K) t}
\end{align}
and we can equivalently simulate $f(\rho_1',\dotsc,\rho_K')$ with $\rho_j'=\proj{1} \otimes \rho_j$ instead.
\end{remark}

\section{Better Phase Estimation}\label{sec:phase_est}

We note that principal component analysis can be performed
using fewer samples than what is claimed in \cite{LMR14}.

\begin{corollary}
Kitaev's phase estimation on the unitary $U=e^{-i\rho}$ can be performed to precision $\epsilon$ and constant failure probability, using $O(1/\epsilon^2)$ samples. 
\end{corollary}

In \cite{LMR14}, they state that this phase estimation requires $O(1/\epsilon^3)$ samples, so this is a polynomial improvement. 

\begin{proof}
Notice that to estimate an eigenvalue of $U$ to precision $\epsilon$ using
standard Kitaev's phase estimation requires $O(1/\epsilon)$ uses of controlled-$U.$
Then as long as the simulation of controlled-$U$ does not change the resulting
state by trace distance more than $O(\epsilon)$, the total error in trace
distance of the final state will be $O(1).$ Using the LMR protocol, we can
simulate $e^{-i\rho}$ to precision $O(\epsilon)$ a total of $O(1/\epsilon)$
times, giving a sample complexity of $O(1/\epsilon^2)$.
\end{proof}

\section{Equivalence of Hermitian polynomials and the Jordan-Lie algebra}\label{sect:Jordan}

In this appendix we sketch that any Hamiltonian that is a Hermitian multinomial in $\rho_1,\rho_2,\dots,\rho_k$ can be created from sums of nested commutators (multiplied by $i$) and anticommutators of density matrices (i.e.\ is in the Jordan-Lie algebra~\cite{Emch} generated by the states), and vice versa.

First, we begin by noticing that for $z\in\mathbb{C}$,
\begin{align}
z \rho_1\rho_2&=\frac{z}{2}\left(\{\rho_1,\rho_2\}+[\rho_1,\rho_2]\right)\\
z \rho_1\rho_2\rho_3&=\frac{z}{4}\left(\{\{\rho_1,\rho_2\}+[\rho_1,\rho_2],\rho_3\}+[\{\rho_1,\rho_2\}+[\rho_1,\rho_2],\rho_3]\right)\\
\dots
\end{align}
So we can write all monomials as sums of nested commutators (but no $i$) and anticommutators.

Now we just want to show that a monomial plus its Hermitian conjugate can be written as a sum of nested commutators (with $i$) and anticommutators. This is possible by noticing that the Hermitian conjugate of an expression of nested commutators and anticommutators (e.g.~$[[\{[A,B],C\},D],E]$) of Hermitian matrices (e.g.~$A,B,C,D,E$) is equal to that same expression with a $(-1)^c$ sign, where $c$ is the number of commutators (alternatively, the number of ``$[$'' symbols) in the expression (e.g.~$[[\{[A,B],C\},D],E]^\dag=(-1)^3[[\{[A,B],C\},D],E]$).

With this fact we can treat all Hermitian polynomials. We replace each monomial by a nested expression of commutators and anticommuatators, and then group together terms with the same parity $c$ of commutators. The terms with even $c$ will contribute to the real part of $z$ while the terms with odd $c$ will contribute to the imaginary part of $z$ (we also need to introduce extra minus signs in front of nested commutators when $c$ is equal to $2$ or $3$ modulo $4$). For instance, we can rewrite the degree-3 monomial plus its Hermitian conjugate as follows:
\begin{align}
z \rho_1\rho_2\rho_3+(z \rho_1\rho_2\rho_3)^\dag&=\frac12\Re(z)\left(\{\{\rho_1,\rho_2\},\rho_3\}-i[i[\rho_1,\rho_2],\rho_3]\right)\\&+\frac12\Im(z)\left(\{i[\rho_1,\rho_2],\rho_3\}+i[\{\rho_1,\rho_2\},\rho_3]\right).
\end{align}

\end{document}